\documentclass[11pt,twoside]{article}
\usepackage[left=1in,top=1in,right=1in]{geometry}
\usepackage[utf8]{inputenc}
\usepackage[T1]{fontenc}
\usepackage{lmodern}
\usepackage[english]{babel}
\usepackage{fancyhdr}
\usepackage{amsmath}
\usepackage{amssymb}
\usepackage{amsfonts}
\usepackage{mathrsfs}
\usepackage{amsthm}
\usepackage{epsfig}
\usepackage{enumerate}
\usepackage[section]{placeins}
\usepackage{caption}
\usepackage{subcaption}
\usepackage[round]{natbib}
\usepackage[svgnames]{xcolor}
\usepackage{hyperref}
\hypersetup{
    unicode=false,          
    pdftoolbar=true,        
    pdfmenubar=true,        
    pdffitwindow=false,     
    pdfstartview={FitH},    
    colorlinks=true,       
    linkcolor=red,          
    citecolor=blue,        
    filecolor=magenta,      
    urlcolor=green           
}
\newcommand{\mytitle}{\textbf{IGS: an IsoGeometric approach for Smoothing on surfaces }}
\newcommand{\scal}[2]{\left\langle {#1},{#2} \right\rangle} 

\newcommand{\norm}[1]{\left\Vert #1\right\Vert}
\newcommand{\trace}{\text{trace}}
\newcommand{\argmin}[1]{\underset{#1}{\operatorname{arg}\,\operatorname{min}}\;}

\newcommand{\dblue}[1]{\textcolor{black}{#1}}

\DeclareFontFamily{OT1}{pzc}{}
\DeclareFontShape{OT1}{pzc}{m}{it}{<-> s * [1.10] pzcmi7t}{}
\DeclareMathAlphabet{\mathpzc}{OT1}{pzc}{m}{it}

\newtheorem{lemma}{Lemma}[section]

\newtheorem{theorem}{Theorem}[section]
\newtheorem{proposition}{Proposition}[section]

\numberwithin{equation}{section}

\begin{document}
\title{\mytitle}

\author{
Matthieu Wilhelm ${}^{1}$\footnote{Corresponding author. E-mail: matthieu.wilhelm@unine.ch 
Phone: +41~32~7181971.},\ Luca Ded\`e ${}^{2}$ ,\ Laura M. Sangalli${}^{3}$, \ Pierre Wilhelm${}^{4}$ \\
{}
${}^{1}$  \ \normalsize Institut de Statistique, 
\normalsize Universit\'e de Neuch\^atel\\[-0.1cm]
\normalsize Avenue de Bellevaux 51, 2000 Neuch\^atel, Switzerland\\[-0.1cm]
{}
${}^{2}$ \ \normalsize Chair of Modeling and Scientific Computing, 
\normalsize \'Ecole Polytechnique F\'ed\'erale de Lausanne\\[-0.1cm]
\normalsize Av. Piccard, Station 8, 1015 Lausanne, Switzerland\\[-0.1cm]
{}
${}^{3}$ \ \normalsize MOX, Dipartimento di Matematica, 
\normalsize Politecnico di Milano\\[-0.1cm]
\normalsize  Via Bonardi 9, 20133 Milano, Italy\\[-0.1cm]
{}
${}^{4}$ \ \normalsize S3,  Swiss Space Systems\\[-0.1cm]
\normalsize Zone Industrielle la Palaz A3, 1530 Payerne, Switzerland \\[-0.1cm]
{}
}
\maketitle

\begin{center}
\begin{minipage}{0.85\linewidth}
\begin{center}
{\large \textbf{Abstract}}
\end{center}
\noindent
We propose \dblue{a novel} approach for smoothing on surfaces, namely estimating a function starting from noisy and discrete measurements.
More precisely, we aim at estimating functions lying on a surface represented by NURBS, which are geometrical representations commonly used in industrial applications. The estimation is based on the minimization of a penalized least-square functional. The latter is equivalent to solve a 4th-order Partial Differential Equation (PDE). In this context, we use Isogeometric Analysis (IGA) for the numerical approximation of such surface PDE, leading to an IsoGeometric Smoothing (IGS) method for fitting data spatially distributed on a surface. Indeed, IGA facilitates encapsulating the exact geometrical representation of the surface in the analysis and also allows the use of at least globally $C^1-$continuous NURBS basis functions for which the 4th-order PDE can be solved using the standard Galerkin method. We show the performance of the proposed IGS method by means of numerical simulations and we apply it to the estimation of the pressure coefficient, and associated aerodynamic force on a winglet of the SOAR space shuttle.

\end{minipage}
\end{center}
\section{Introduction}
The estimation of a function from a set of noisy data is a very common task, which is often tackled by minimization of a penalized least-square functional, where the penalty involves a differential operator, commonly based on second order derivatives, and occasionally on derivatives of a different order. Classical examples are offered by smoothing splines (see, e.g., \citet{Silverman05} and references therein) for estimating functions defined over real intervals, thin-plate splines (see, e.g., \citet{Wahba90}, \citet{Wood06} and references therein), Soap film smoothing \citep{Wood08}, splines over triangulations (see, e.g., \citet{Lai07}) for estimating functions defined over regions of $\mathbb{R}^2$, and spherical splines (see, e.g., \citet{Wahba81, Baramidze06, Alfeld96}) for estimating functions defined over spheres or spheres-like surfaces. In this context, one of the main challenges in  minimizing such penalized least-square functional consists in determining a suitable finite dimensional space representing, at the discrete level, the infinite dimensional space to which the function belongs. In other words, the challenge is to find a finite dimensional problem which is tractable and whose solution is close to the solution in the infinite dimensional space. The same challenge arises when looking for the numerical solution of a PDE. This common goal has been recently exploited to develop statistical tools to deal with spatially distributed data. In this respect, \citet{Ramsey02} considered planar smoothing optimizing a penalized least-square functional with a regularization term involving the Laplacian, and used the Finite Element Method (FEM)  (see, e.g., \citet{Quarteroni94}) to solve the estimation problem. \citet{Sangalli13} generalized the method proposed by \citet{Ramsey02} to include space-varying covariates and to account for boundary conditions, while \citet{Wilhelm13} explored a generalized linear version of the method.
\citet{Azzimonti1,Azzimonti2} further extended the model of \citet{Sangalli13} to account for any elliptic differential penalization, not necessarily based on the Laplacian operator.
 \citet{Ettinger12} and \citet{Dassi15} dealt with the case of functions defined over two dimensional Riemannian manifolds, by considering a regularizing term based on the Laplace-Beltrami operator and exploited a conformal parametrization of the manifold to solve an equivalent estimation problem on $\mathbb{R}^2$. \citet{Duchamp03} also explored smoothing on complex surfaces using a penalization based on the Laplace-Beltrami operator. A common feature of all these contributions is the use of FEM to estimate the underlying function corresponding to the observed data.

In this paper, we consider the estimation of functions defined over surfaces in $\mathbb{R}^3$ starting from a discrete set of noisy observed data in points distributed on such surfaces. More precisely, we refer to sufficiently smooth functions defined over surfaces that can be represented by NURBS (Non-Uniform Rational Basis Splines). Indeed, NURBS are commonly used in Computer Aided Design (CAD) to represent most of the geometries of Engineering and industrial interest. From a more general point of view, the proposed model is a particular case of a Generalized Additive Model (GAM) \citep{Hastie90} where the smooth component is defined on a surface. Hence, most of the theoretical results of GAMs can be directly applied to this model, e.g., for the quantification of uncertainty.

IGA has been first introduced by \citet{Hughes05} with the main idea of using the same basis functions to represent the geometry and then to approximate the solution of the PDEs defined in such computational domains. This facilitates encapsulating the exact geometrical representation when performing the analysis of PDEs defined in the computational domain; see \citet{Cottrell09}. 
In this respect, the isogeometric paradigm facilitates the use of an exact representation of the surface, while most of the current methodologies, as FEM, generally only handle its approximation. This may induce an error on the solution due to the approximation of the geometry and may require complex meshing procedures. As mentioned, to estimate the function over the surface, starting from its discrete and noisy measurements, we minimize a least-square functional where the penalty involves the Laplace-Beltrami operator associated to the surface, analogously to \citet{Ettinger12}, \citet{Duchamp03} and \citet{Wahba81}. The estimation problem is tackled directly on the surface and using IGA to numerically solve the associated surface PDE. For this reason, we name the resulting method IsoGeometric Smoothing (IGS). High-order PDEs can be straightforwardly solved by NURBS-based IGA; indeed, globally $C^k-$continuous NURBS basis functions can be easily defined on surfaces for some $k=0,\dots,p-1$, where $p$ is the polynomial degree. This allows to use standard Galerkin method for numerically solving the PDE. In this respect, \citet{Dede1} studied the use of IGA for 2nd-order PDEs defined on surfaces, \citet{Dede2} analysed IGA for high-order PDEs, while \citet{Bartezzaghi15} for high-order PDEs defined on surfaces. All these works showed the efficiency and accuracy of IGA for the spatial approximation of surface PDEs.  \dblue{Partially related to our work, NURBS-based and IGA approaches are used in \citet{reviewer1} and in \citet{reviewer2} as calibration tools.
}
\begin{figure}[t]
\centering
\includegraphics[scale=.5]{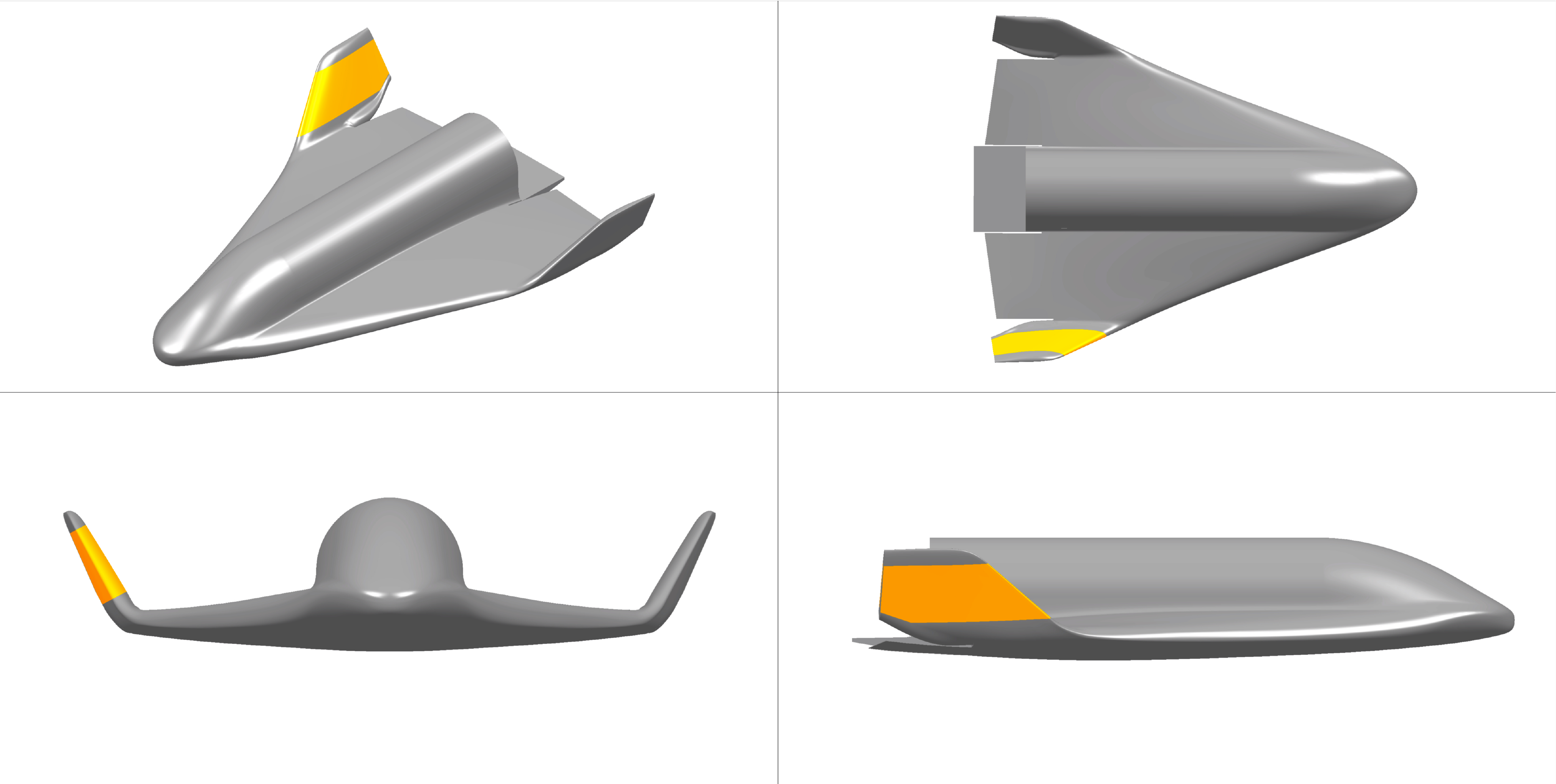}
\caption{\label{fig::shuttle} The SOAR shuttle highlighting the  inboard winglet represented by a single NURBS patch. [Courtesy of \textit{S3, Swiss Space Systems Holding SA}].}
\end{figure}

In this paper, we introduce the geometrical background to define the application framework of IGS. Then, we show how one can minimize a least-square functional involving a penalization term and hence estimate a function starting from a set of discrete and noisy measurements. We assess and compare IGS to the Thin Plate Splines (TPS) \citep{Duchon77, Wahba90} by means of numerical simulations. Finally, as industrial application, we estimate the pressure coefficient and the aerodynamic force acting on the inboard winglet of the space shuttle SOAR, designed by \textit{S3, Swiss Space Systems Holding SA}. S3  is a Swiss company currently developping, manufacturing, certifying, and operating a launch system for small satellites of weight inferior to 250 kg. A geometry of the SOAR suborbital shuttle for a preliminary study of aerodynamic forces is shown in Figure~\ref{fig::shuttle}. In this application, the data are pointwise measurements of the pressure coefficient and the associated quantity of interest is the aerodynamic force. We propose a method to estimate functions defined over a NURBS surface starting from scattered noisy observations. 

This paper is organized as follows. In Section \ref{sec::NURBS}, we briefly recall B-splines and NURBS. In Section \ref{sec::frame_model}, we introduce some analytical tools and the IGS method. We show results for two numerical simulations in Section \ref{sec:num}: the first one allows to compare IGS to TPS, while the second one corresponds to a complex NURBS surface, for which planar smoothers are not suited. Finally, in Section \ref{sec:soar},  we show the results of the estimation of the pressure coefficient and aerodynamic force over the winglet of the SOAR space shuttle. Conclusions follow.
\section{B-splines and NURBS}
\label{sec::NURBS}
NURBS are widely used in CAD for geometrical representation of surfaces \citep{Piegl97}. We start by defining an \emph{open knot vector} of degree $p$ as a \dblue{sequence} of values $\varXi= \left\{\xi_1, \dots, \xi_{n+p+1} \right\}$, with knots $\xi_1\leq \ldots \leq \xi_{n+p+1}$, where the first and the last knots are repeated $p+1$ times. The interior knots can be repeated at most $p$ times and, if a knot is repeated $m$ times, we say that its multiplicity is $m$. The B-spline basis functions $N_{i,p}$ are defined recursively using the Cox-De Boor formula.
We summarize some properties of B-spline basis functions $N_{i,p}$ of index $i$ and degree $p$, for $i=1,\dots, n$.
\begin{itemize}
\item The basis function $N_{i,p}$ possesses $p-m$ continuous derivatives across a knot of multiplicity $1\leq m\leq p$ and is $C^\infty-$ continuous between the knots.
\item The support of the basis function $N_{i,p}$ is compact and contained in $p+1$ knots spans $[\xi_i, \xi_{i+p+1}]$.
\item The basis functions are pointwise nonnegative, i.e. $N_{i,p}(\xi)\geq 0,$ for all $ i=1,\dots,n$.
\item They form a partition of the unity, that is $\displaystyle \sum_{i=1}^n N_{i,p}(\xi)=1$, for all $\xi \in [\xi_1, \xi_n]$.
\end{itemize}
Starting from the basis function $N_{i,p}$, the NURBS basis functions $R_{i,p}(\xi)$ are defined as projective transformations of B-spline basis functions. Let $w_1,\dots,w_n >0$ be positive weights, then, NURBS basis functions are defined as:
$$R_{i,p}(\xi)= \frac{w_i}{\sum_{j=1}^n w_j N_{j,p}(\xi)} N_{i,p}(\xi), \quad \forall i=1,\dots, n.$$
A NURBS curve $C(\xi) \in \mathbb{R}^d$, with $d=2,3$ and control points $\mathbf{B}_1,\dots,\mathbf{B}_n \in \mathbb{R}^d$ is defined as:
$$C(\xi)= \sum_{i=1}^n R_{i,p}(\xi)\mathbf{B}_i.$$
\begin{figure}[t]
\begin{subfigure}{\linewidth}
\centering
\includegraphics[scale=0.41]{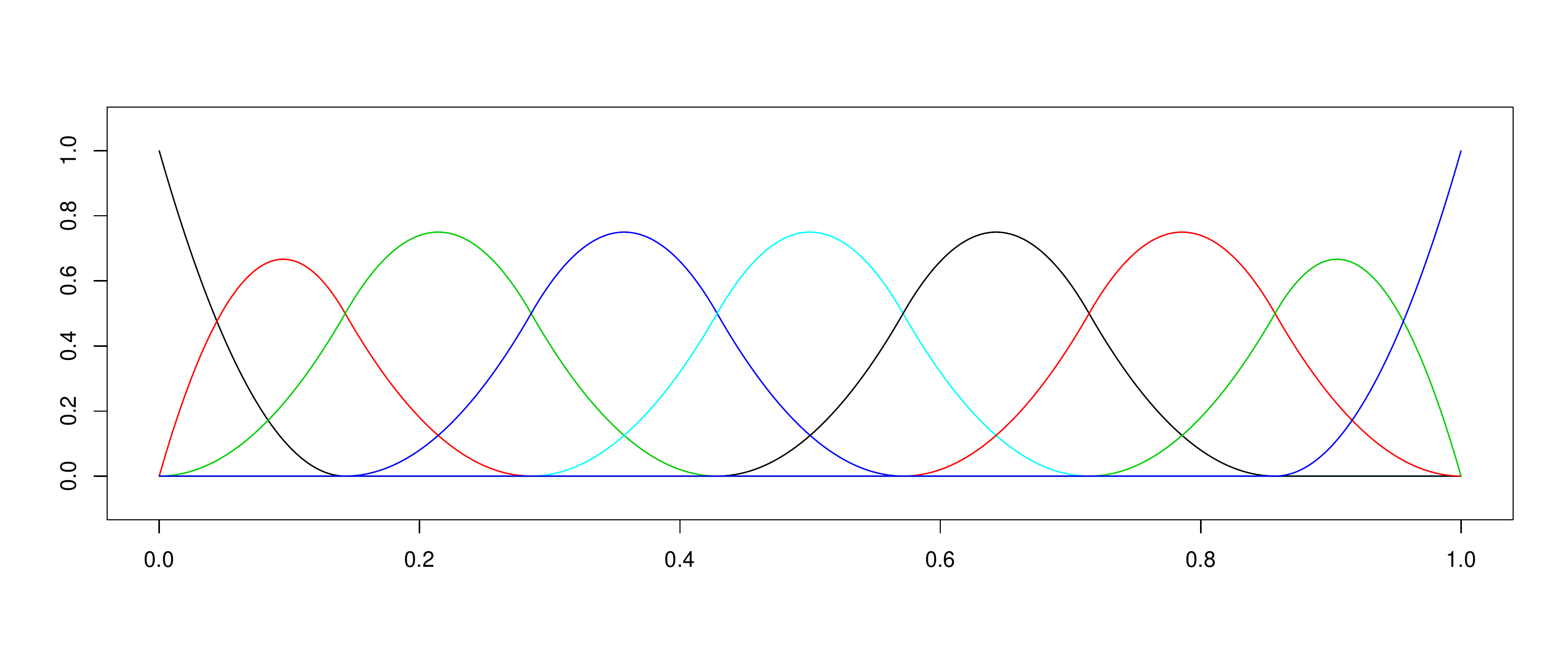} 
\end{subfigure}
\begin{subfigure}{\linewidth}
\centering
\includegraphics[scale=0.41]{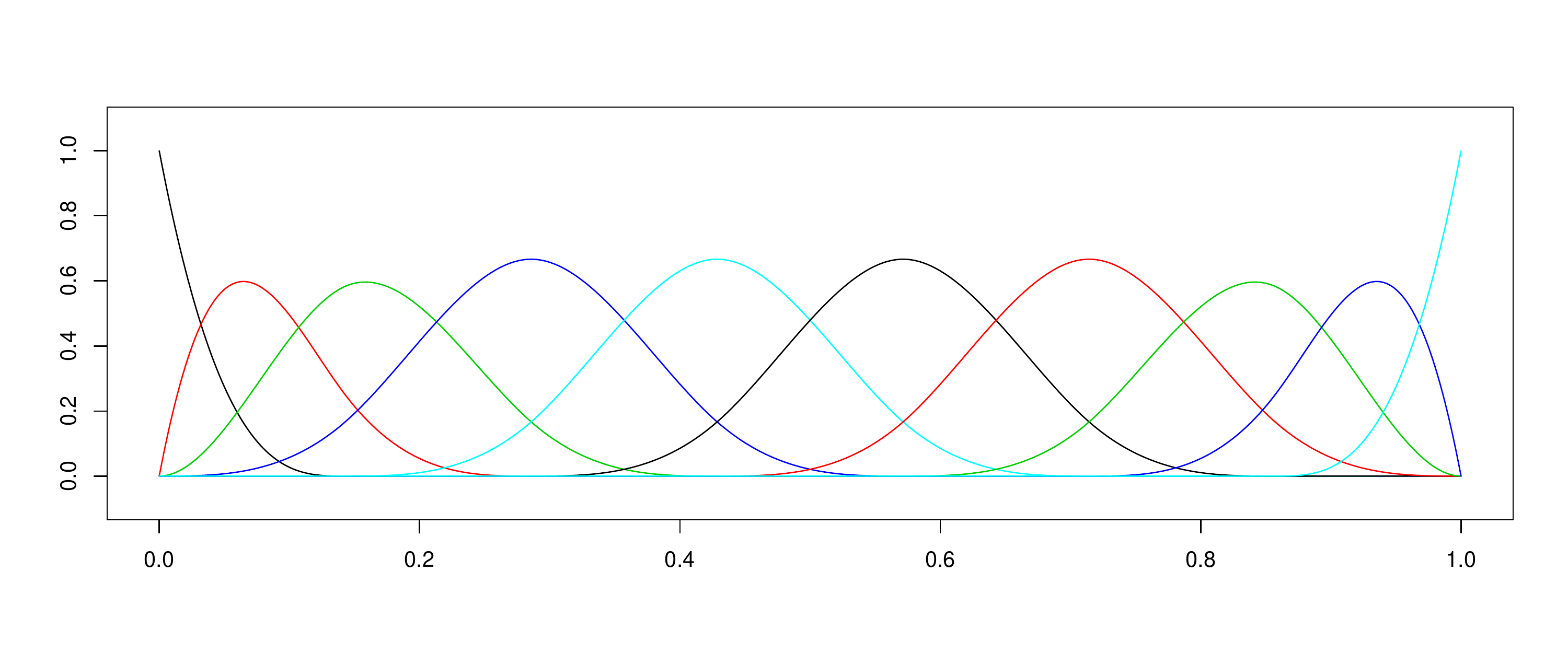} 
\end{subfigure} 
\caption{\dblue{B-spline basis functions  for different knots vectors. Basis functions of degree $p=2$ and globally $C^1-$continuous (top) and of degree $p=3$ and globally $C^2-$continuous (bottom).}}
\end{figure}

To define surfaces with NURBS, we resort to the tensor product \dblue{scheme.} Given two open knot vectors $\varXi= \left\{\xi_1, \dots, \xi_{n+p+1} \right\}$ and $\mathscr{H} =\left\{\eta_1, \dots, \eta_{m+q+1} \right\}$, $Q_{i,p}(\xi)$, for $i=1,\dots,n$, $M_{j,q}(\eta)$, for $j=1,\dots,m$, the corresponding univariate
\dblue{
basis functions, the control points $\mathbf{B}_{ij}\in \mathbb{R}^d,$ and positive weights $w_{ij}$ for $ i=1,\dots, n,\ j=1,\dots,m,$ we define the bivariate NURBS basis functions as: 
$$ R_{i,j}(\xi, \eta) =\frac{Q_{j,p}(\xi)M_{k,q}(\eta) w_{ij}}{\sum_{k=1}^n\sum_{l=1}^m Q_{k,p}(\xi)M_{l,q}(\eta) w_{kl}}, $$
and a NURBS surface as:
$$
S: \Omega \rightarrow \mathbb{R}^d,\quad 
 S(\xi, \eta)=\sum_{i=1}^{n} \sum_{j=1}^{m} R_{i,j}(\xi, \eta)\mathbf{B}_{ij},
$$
where $\Omega = (\xi_1,\xi_m)\times (\eta_1,\eta_m)$. For simplicity, we consider henceforth the same polynomial degree $p$ along both the parametric directions, for which we rewrite the bivariate basis functions simply as $R_{i,p}(\xi, \eta) $, for $i = 1,\dots, N^h,$ where $N^h = n m $.} For an exhaustive description of NURBS, we refer the reader to \citet{Piegl97}. 

We remember that in IGA, one uses the same basis functions to represent the surface and to approximate the solution of the PDE defined in such computational domain. In general, it is possible to enrich the NURBS basis without changing the geometry, \dblue{i.e. by preserving the geometrical mapping, with the goal of obtaining a more accurate solution of the PDEs. In this respect, $h-$refinement indicates a uniform knot insertion, which adds new basis functions while preserving the geometrical mapping, while $p-$refinement refers to an elevation of the polynomial degree of the basis functions, similarly to FEM.} In addition, the so-called $k-$ refinement is peculiar of NURBS and refers to a consecutive order elevation and a knot insertion which allow to increase the polynomial degree and continuity of basis functions. All the refinement procedures are discussed in details in \citet{Cottrell09}, while for an introduction to NURBS in the context of IGA, we refer  the interested reader to \citet{Hughes05}. 
\section{The IGS method}
\label{sec::frame_model}
We describe the mathematical framework of the IGS method. First, we introduce the surface differential operators and then we introduce the IGS method for smoothing functions on surfaces.
\subsection{Geometrical framework}
Following \citet{Dede1}, let $\Omega \subset\mathbb{R}^2$ be an open and bounded parametric domain of finite measure with respect to the topology of $\mathbb{R}^2$. Then, let $\Sigma\subset \mathbb{R}^3$ be a compact, connected, and oriented surface, defined by a NURBS geometrical mapping $X:\Omega\rightarrow\Sigma$ such that: 
\begin{equation}
\label{equ::mapping}
X:\Omega\subset \mathbb{R}^2 \rightarrow\Sigma\subset \mathbb{R}^3,\quad
\mathbf{s}=(s_1,s_2)\mapsto\mathbf{x}=(x_1,x_2,x_3).
\end{equation}
We assume that $X$ is sufficiently smooth, e.g. $X\in C^1(\Omega)$. Then, we define the Jacobian of the mapping $X$, denoted by $\nabla X$ as:
$$\nabla X: \Omega\rightarrow \mathbb{R}^{3\times 2}, \quad \mathbf{s}\mapsto \nabla X(\mathbf{s}),\quad (\nabla X)_{i,j}(\mathbf{s})=\frac{\partial X_i}{\partial \mathbf{s}_j}(\mathbf{s}), \quad i=1,2,3, \ j=1,2.$$
We denote by $(\nabla X)_i$ the $i$-th column of the matrix $\nabla X$.
\dblue{The metric tensor of the mapping $X$ is represented by:}
$$\mathbf{G}: \Omega \rightarrow \mathbb{R}^{2\times2}, \quad \mathbf{s}\mapsto \mathbf{G}(\mathbf{s}), \quad \mathbf{G}(\mathbf{s})= \left(\nabla X(\mathbf{s})\right)^T \nabla X(\mathbf{s})$$
and we denote with $g(\mathbf{s})$ the square root of the determinant of \dblue{the metric tensor} of the mapping $X$:
$$g: \mathbf{\Omega}\rightarrow\mathbb{R}, \quad \mathbf{s}\mapsto g(\mathbf{s}), \quad g(\mathbf{s})=\sqrt{\det{(\mathbf{G}(\mathbf{s}))}}.$$
\dblue{The metric tensor} of the mapping $X$ is assumed to be invertible almost everywhere in $\Omega$.
\subsection{Differential operators}
Let $\phi \in C^2(\Sigma) $ be a smooth function defined over the surface $\Sigma$, i.e.: 
 $$\phi:\Sigma\subset \mathbb{R}^3 \rightarrow \mathbb{R},\quad \mathbf{x}\mapsto \phi(\mathbf{x}),\quad \phi\in C^2(\Sigma),$$
such that we can define the differential operators on the surface $\Sigma$. Then, the projection operator $\mathbf{P}(\mathbf{x})$ on the \dblue{tangent} plane to the surface at the point $\mathbf{x}\in \Sigma$ is given by:
$$\mathbf{P}(\mathbf{x})=(\mathbf{I}-\mathbf{n}_\Sigma\otimes\mathbf{n}_\Sigma)(\mathbf{x})= \mathbf{I}- \mathbf{n}_\Sigma(\mathbf{x})\mathbf{n}_\Sigma(\mathbf{x})^T, \quad \forall \mathbf{x}\in \Sigma,$$
where $\mathbf{n}_\Sigma(\mathbf{x})$ is the unit normal vector to the surface at the point $\mathbf{x}$ and $\mathbf{I}$ is the identity matrix\footnote{We compute $\mathbf{n}_\Sigma(\mathbf{x})$ as $\mathbf{n}_\Sigma(\mathbf{x})=\frac{[\nabla X]_1(\mathbf{s}) \times [\nabla X]_2(\mathbf{s})}{\left\|[\nabla X]_1(\mathbf{s}) \times [\nabla X]_2(\mathbf{\mathbf{s}})\right\|_2}$, where $\mathbf{s}= X^{-1}(\mathbf{x})$.}.
Then, the surface gradient operator, denoted by $\nabla_{\Sigma}$, is defined as the projection of the gradient of a function extended in a tubular region containing $\Sigma$, i.e.:
$$\nabla_{\Sigma}\phi(\mathbf{x}):= \mathbf{P}(\mathbf{x})\nabla \phi(\mathbf{x})=\nabla \phi(\mathbf{x})- (\mathbf{n}_\Sigma(\mathbf{x})^T \nabla \phi(\mathbf{x}))\mathbf{n}_\Sigma(\mathbf{x}).$$
The Laplace-Beltrami operator, which is the surface Laplacian operator, is expressed as:
$$\Delta_{\Sigma} \phi(\mathbf{x})= \nabla_\Sigma\cdot(\nabla_\Sigma \phi(\mathbf{x}))= \trace[\mathbf{P}(\mathbf{x})\, \nabla^2 \phi \, \mathbf{P}(\mathbf{x})],\quad \forall \mathbf{x}\in \Sigma,$$
where $\nabla_\Sigma \cdot \mathbf{v}(\mathbf{x})= \trace(\nabla_\Sigma \mathbf{v}(\mathbf{x}))$, for all $ \mathbf{v}\in C^1(\Sigma)$, is the surface divergence and \newline $(\nabla^2 \phi)_{ij}(\mathbf{x})=\frac{\partial^2 \phi}{\partial x_i\partial x_j}(\mathbf{x})$ is the Hessian matrix of $\phi$. These operators can be rewritten in the parametric domain $\Omega$ using the geometrical mapping (\ref{equ::mapping}) as:
\begin{equation}
\label{equ::diff_operators}
\nabla_\Sigma \phi(\mathbf{x})= \nabla X(\mathbf{s}) \mathbf{G}^{-1}(\mathbf{s})\nabla (\phi\circ X)(\mathbf{s})\text{ and }
 \Delta_\Sigma \phi(\mathbf{x})= \frac{1}{g(\mathbf{s})} \nabla\cdot \left[ g(\mathbf{s}) \mathbf{G}^{-1}(\mathbf{s})\nabla (\phi\circ X)(\mathbf{s})\right],
\end{equation}
where $\mathbf{s}=X^{-1}(\mathbf{x})$.
\subsection{Mathematical model}
\label{sec::math_model}
Let us consider $N$ points $\mathbf{p}_1,\dots,\mathbf{p}_N$ located on $\Sigma$ for which the observed values are $y_1,\dots,y_N$ respectively. We assume that:
\begin{equation}
\label{equ::math_mod} 
y_i=f(\mathbf{p}_i)+\varepsilon_i, \quad \forall\  i=1,\dots, N,
\end{equation}
where $f$ is a sufficiently smooth function defined on $\Sigma$ that we aim at estimating and $\varepsilon_i$ are independent observational errors of zero mean and constant variance.

Given a positive smoothing parameter $\lambda$, we aim at minimizing the following parameter dependent functional:
\begin{equation}
\label{equ::functional}
J_{\lambda}(v)=\sum_{i=1}^N (y_i-v(\mathbf{p}_i))^2+\lambda \int_\Sigma (\Delta_\Sigma v)^2\ d\Sigma= \norm{\mathbf{y}-\mathbf{v}_N}_2^2+\lambda \int_\Sigma (\Delta_\Sigma v)^2\ d\Sigma,
\end{equation}
where $\mathbf{y}=(y_1,\dots,y_N)^T$ and $\mathbf{v}_N=(v(\mathbf{p}_1),\dots,v(\mathbf{p}_N))^T$ is the vector of evaluations of the general function $v$ at the points $\mathbf{p}_1,\dots, \mathbf{p}_N$. This functional is the same considered by \citet{Ettinger12} and by \citet{Duchamp03} that uses a finite element representation, and by \citet{Wahba81}, that considers functions on spheres and proposes spherical splines. The functional (\ref{equ::functional}) involves the surface Laplace-Beltrami operator, which is, roughly speaking, a measure of the curvature of the function related to the surface. 
The use of the Laplace-Beltrami operator ensures that the regularization is invariant to rigid transformations of the coordinate system, which is desirable both for theoretical and practical reasons.
If the observational errors are normally distributed, the functional (\ref{equ::functional}) can be viewed as a negative rescaled Gaussian penalized log-likelihood. Hence, in such case, minimizing (\ref{equ::functional}) is equivalent to maximizing a penalized log-likelihood. Then, for a given positive parameter $\lambda$, the estimation problem is:
\begin{equation}
\label{equ::min_functional}
\text{find } \hat{f} \in \mathcal{F}:\  J_\lambda(\hat{f})\leq J_\lambda(v), \quad \forall v\in \mathcal{F},
\end{equation}
where $\mathcal{F}$ is a suitable functional space to be defined to ensure the well-posedness of the problem and $\hat{f}$ is the estimate of $f$ in the functional space $\mathcal{F}$, which should lie at least in $H^2(\Sigma)$, i.e. $\mathcal{F}\subset H^2(\Sigma)$. Indeed, since $\hat{f}\in H^2(\Sigma)\subset C^0(\Sigma)$, the evaluation of $\hat{f}$ at the points $\mathbf{p}_1,\dots,\mathbf{p}_N $ is well defined. Since the problem (\ref{equ::min_functional}) will be later associated to a PDE, some essential boundary conditions \citep{Brezis99} should be specified in relation with the choice of $\mathcal{F}$. In the simpler context of functions defined over planar domains, and with the penalizing terms involving linear second order elliptic operators, the well-posedness of problem (\ref{equ::functional}) is extensively discussed in \citet{Azzimonti1} under different kind of boundary conditions (Dirichlet, Neumann, or Robin (see, e.g., \citet{Brezis99}). In the following, we prove the well-posedness of the estimation problem (\ref{equ::min_functional}) in the particular case of homogeneous boundary conditions using the Lax-Milgram theorem, following \citet{Ramsey02} and \citet{Sangalli13}.
Let $H^2_0(\Sigma)$ be defined as:
\begin{equation}
\label{equ::H20}
H^2_0(\Sigma):= \left\{v \in H^2(\Sigma)\ : \ \nabla_\Sigma v \cdot \mathbf{n} = 0 \text{ and }  v= 0  \text{ on } \partial\Sigma \right\},
\end{equation}
where $\mathbf{n}$ denotes the outward directed unit vector normal to $\partial\Sigma$, the boundary of $\Sigma$. Then, one can characterize the solution of the minimization problem (\ref{equ::min_functional}) and ensure the existence and the uniqueness of the solution in the case where $\hat{f}$ is assumed to lie in $\mathcal{F}= H^2_0(\Sigma)$.
With this aim, we recall the Lax-Milgram Theorem (see, e.g., \citet{Quarteroni94}): 
\begin{theorem}[Lax-Milgram]
Let $\mathcal{F}$ be a Hilbert space, $G(\cdot,\cdot):\mathcal{F}\times\mathcal{F}\rightarrow\mathbb{R}$ a continuous and coercive bilinear form and $F:\mathcal{F}\rightarrow\mathbb{R}$ a linear and continuous functional. Then, there exists a unique solution of the following problem:$$ \text{find }u\in \mathcal{F} \text{ :  }G(v,u)=F(v),\quad \forall v\in \mathcal{F}.$$
Moreover, if $G(\cdot,\cdot)$ is symmetric, then $u \in \mathcal{F}$ is the unique minimizer in $\mathcal{F}$ of the functional $J:\mathcal{F}\rightarrow \mathbb{R}$, defined as
$$J(v)=G(v,v)-2F(v).$$
\end{theorem}
\begin{lemma}
\label{lem::bilin_form}
Let $v,u\in \mathcal{F}= H^2_0(\Sigma)$ be two functions and $\lambda>0$ a positive smoothing parameter. Let $G(\cdot,\cdot)$ and $F(\cdot)$ be defined as:
\begin{equation}
\label{equ::def_lin_bilin_forms}
G(v,u)= \scal{\mathbf{v}_N}{\mathbf{u}_N}+\lambda \int_\Sigma (\Delta_\Sigma v)(\Delta_\Sigma u) \ d\Sigma, \qquad \text{and}\qquad F(v)=\scal{\mathbf{y}}{\mathbf{v}_N},
\end{equation}
where $\mathbf{v}_N = \left(v(\mathbf{p}_1),\dots,v(\mathbf{p}_N)\right)^T$ and $\mathbf{u}_N = \left(u(\mathbf{p}_1),\dots,u(\mathbf{p}_N)\right)^T$ for some distinct points $\mathbf{p}_1,\dots,\mathbf{p}_N$ on $\Sigma$. Then, the bilinear form $G(\cdot,\cdot)$ is coercive, continuous and symmetric and $F(\cdot)$ is linear and continuous in $H^2_0(\Sigma)$.
\end{lemma}

\begin{proof}[Proof]
First, we note that $|\cdot|_{H^2(\Sigma)}$, defined as $|v|_{H^2(\Sigma)}= \int_\Sigma (\Delta_\Sigma v)^2 \ d\Sigma$, is equivalent to the norm $\|\cdot\|_{H^2(\Sigma)}$ in $H^2_{0}(\Sigma)$. This implies that there exists a constant $C_{0,\Omega}>0$ such that $ |v|_{H^2(\Sigma)}\geq C_{0,\Omega} \|v\|_{H^2(\Sigma)}, \ \forall v\in H^2_0(\Sigma)$ (see, e.g., \citet{Quarteroni09}). Then, we have:
$$\displaystyle G(v,v) = \scal{\mathbf{v}_N}{\mathbf{v}_N} + \lambda\int_{\Sigma} (\Delta_\Sigma v)^2\ d\Sigma =\|\mathbf{v}_N\|_2^2 + \lambda |v|^2_{H^2(\Sigma)} \geq\lambda C_{0,\Omega} \|v\|^2_{H^2(\Sigma)}, \quad \forall v \in H^2_{0}(\Sigma),
$$
and so $G(\cdot,\cdot)$ is coercive. We now show the continuity of $G(\cdot,\cdot)$. Since $H^2(\Sigma)\subset C^0(\Sigma)$, there exists a constant $C_{1,\Omega, N }$ such that $ \|\mathbf{v}_N\|_2 \leq C_{1,\Omega,N} \| v\|_{H^2(\Sigma)}, \ \forall v \in H^2(\Sigma)$. We have:
$$
\renewcommand{\arraystretch}{2}
\begin{array}{lll}
\displaystyle G(v,u)&=&\displaystyle \scal{\mathbf{v}_N}{\mathbf{u}_N} + \lambda\int_{\Sigma} (\Delta_\Sigma v)(\Delta_\Sigma u) \leq \|\mathbf{v}_N\|_{2}\|\mathbf{u}_N\|_{2}+\lambda |v|_{H^2(\Sigma)}  |u|_{H^2(\Sigma)} \\
&\leq & C_{1,\Omega,N}^2 \|v\|_{H^2(\Sigma)}  \|u\|_{H^2(\Sigma)}  + \lambda |v|_{H^2(\Sigma)}  |u|_{H^2(\Sigma)} \\
&\leq &  \max\{C_{1,\Omega,N}^2, \lambda\} \|v\|_{H^2(\Sigma)} \|u\|_{H^2(\Sigma)}, \quad \forall  v,u \in H^2_{0}(\Sigma).\\
\end{array}
$$
Then, the bilinear form $G(\cdot,\cdot)$ is also continuous and its symmetry is obvious.

Finally, we have:
$$|F(v)|= |\scal{\mathbf{y}}{\mathbf{v}_N}|\leq \|\mathbf{y}\|_{2} \|\mathbf{v}_N\|_2 \leq C_{1,\Omega,N} \|\mathbf{y}\|_{2} \|v\|_{H^2(\Sigma)},\quad \forall v \in H^2_{0}(\Sigma),$$
which proves the continuity of the linear form $F$ and concludes the proof.  
\end{proof}
\begin{proposition} \label{prop::exit_uniq}
Let $\mathcal{F}= H^2_0(\Sigma)$ and $\lambda>0$. Then, the solution of problem (\ref{equ::min_functional}) exists and is unique. Moreover, problem (\ref{equ::min_functional}) is equivalent to:
\begin{equation}
\label{equ::caract_sol_inf_dim}
\text{find $\hat{f} \in \mathcal{F}$ : }\scal{\mathbf{v}_N}{\mathbf{\hat{f}}_N} +\lambda \int_\Sigma \Delta_\Sigma v\, \Delta_\Sigma \hat{f}\ d\Sigma  = \scal{\mathbf{y}}{\mathbf{v}_N},\quad \forall v \in \mathcal{F}.
\end{equation}
\end{proposition}
\begin{proof}[Proof]
The functional $J_\lambda(v)$ (\ref{equ::functional}) can be rewritten as:
$$J_\lambda(v)=\norm{\mathbf{y}-\mathbf{v}_N}_2^2+\lambda \int_\Sigma (\Delta_\Sigma f)^2 \ d\Sigma= \norm{\mathbf{y}}_2^2-2\,\scal{\mathbf{y}}{\mathbf{v}_N}+\norm{\mathbf{v}_N}_2^2 +\lambda\int_\Sigma (\Delta_\Sigma v)^2 \ d\Sigma.$$
By defining $\tilde{J}_\lambda(v)$ as
$$\tilde{J}_\lambda(v)=\norm{\mathbf{v}_N}_2^2+\lambda \int_\Sigma (\Delta_\Sigma v)^2\ d\Sigma -2\,\scal{\mathbf{y}}{\mathbf{v}_N},$$
we have that $$\argmin{v\in \mathcal{F}} J_\lambda(v) \equiv \argmin{v\in \mathcal{F}} \tilde{J}_\lambda(v).$$
From the definitions of $G(v,u)$ and $F(v)$  of (\ref{equ::def_lin_bilin_forms}), the functional $\tilde{J}_\lambda(v)$
can be written as:
$$\tilde{J}_\lambda(v)=G(v,v)-2F(v).$$
Thanks to the Lax-Milgram Theorem, it is then sufficient to use Lemma \ref{lem::bilin_form} to establish the well-posedness of problem (\ref{equ::min_functional}). Moreover, since $G(\cdot, \cdot)$ is symmetric, problem (\ref{equ::min_functional}) is equivalent to problem (\ref{equ::caract_sol_inf_dim}).
\end{proof}
Setting a priori the essential boundary conditions $ \nabla_\Sigma \hat{f} \cdot \mathbf{n} = 0 $ and $  \hat{f}= 0  \text{ on } \partial\Sigma$ following (\ref{equ::H20}) may be an inadequate choice in several applications, especially when the behaviour of the function $f$ at the boundaries is not known a priori. In such cases, it may be more convenient to consider instead \emph{natural boundary conditions}, that is:
\begin{equation}
\label{equ:nat_bound_cond}
\left\{
\renewcommand{\arraystretch}{1.5}
\begin{array}{ll}
\nabla_\Sigma(\Delta_\Sigma \ \hat{f})\cdot \mathbf{n} = 0  & \text{ on }\partial\Sigma ,\\
\Delta_\Sigma \ \hat{f}  = 0  & \text{ on }\partial\Sigma.
\end{array}\right. 
\end{equation}
In the case that the boundary conditions (\ref{equ:nat_bound_cond}) are set, we are unable to show the well-posedness of problem (\ref{equ::caract_sol_inf_dim})  with $\mathcal{F}= H^2(\Sigma)$. Nevertheless, numerical experience indicates that it still yields a numerically stable problem. We can also note that the usual planar smoothers, such as TPS, also implicitly use natural boundary conditions.
\subsection{Numerical approximation: IGS}
Let $\hat{f}^h$ be a finite dimensional approximation of $\hat{f}$ obtained by means of IGS.
Let $\left\{\psi_1,\dots,\psi_{N^h}\right\}$ be a basis of a discrete function space $\mathcal{F}^h\subset\mathcal{F}\subseteq H^2(\Sigma)$ of dimension $N^h$. In the finite dimensional space $\mathcal{F}^h$, problem (\ref{equ::caract_sol_inf_dim}) reads:
\begin{equation}
\label{equ::finite_dim_prob}
\text{find } \hat{f}^h\in \mathcal{F}^h \text{ : }  \scal{\mathbf{v}^h_N}{\mathbf{\hat{f}}^h_N} +\lambda \int_\Sigma \Delta_\Sigma v^h\, \Delta_\Sigma \hat{f}^h \ d\Sigma= \scal{\mathbf{y}}{\mathbf{v}^h_N},\qquad \forall v^h \in \mathcal{F}^h,
\end{equation}
where $\mathbf{v}_N^h:=(v^h(\mathbf{p}_1),\dots,v^h(\mathbf{p}_N))^T$ and $\mathbf{\hat{f}}_N^h:=(\hat{f}^h(\mathbf{p}_1),\dots,\hat{f}^h(\mathbf{p}_N))^T$. Since $\mathcal{F}^h$ is finite dimensional, problem (\ref{equ::finite_dim_prob}) is equivalent to:
\begin{equation}
\label{equ::finite_dim_weak_form}
\text{find } \hat{f}^h \text{ : } \scal{\boldsymbol{\psi_i}_N}{\mathbf{\hat{f}}^h_N} + \lambda \int_\Sigma \Delta_\Sigma \psi_i\, \Delta_\Sigma \hat{f}^h\ d\Sigma = \scal{\mathbf{y}}{\boldsymbol{\psi_i}_N},\qquad \forall i=1\dots,{N^h},
\end{equation}
where $\boldsymbol{\psi_i}_N:=(\psi_i(\mathbf{p}_1),\dots,\psi_i(\mathbf{p}_N))^T$. Let us define the ${N^h}\times {N^h}$ matrix $\mathbf{R}$ as\\ $(\mathbf{R})_{ij}=\int_\Sigma \Delta_\Sigma \psi_i \, \Delta_\Sigma \psi_j\ d\Sigma,$
and the $N\times {N^h}$ matrix $\boldsymbol{\Psi}$ as $(\boldsymbol{\Psi})_{ij}= \psi_j(\mathbf{p}_i).$
Since $\hat{f}^h$ belongs to $\mathcal{F}^h$, it can be written as a linear combination of the basis functions:
$$\hat{f}^h(\mathbf{x})=\sum_{i=1}^{N^h} \hat{f}_i\  \psi_i(\mathbf{x}), \quad \forall \mathbf{x}\in \Sigma,$$
or compactly as
$\hat{f}^h(\mathbf{x})=\boldsymbol{\psi}^T(\mathbf{x})\mathbf{\hat{f}}^h,$
where $\mathbf{\hat{f}}^h:=(\hat{f}_1,\dots,\hat{f}_{N^h})^T$ and \\ $\boldsymbol{\psi}(\mathbf{x})=(\psi_1(\mathbf{x}),\dots,\psi_{N^h}(\mathbf{x}))^T$. Then, we have:
$$\mathbf{\hat{y}}:=\mathbf{\hat{f}}_N^h=\boldsymbol{\Psi}\mathbf{\hat{f}}^h.$$
Problem (\ref{equ::finite_dim_weak_form}) in matrix form reads as:
\begin{equation}
\label{equ::problem_mat_form}
\text{find }\mathbf{\hat{f}}^h\in \mathbb{R}^{N^h} \text{:} \qquad \mathbf{A}\mathbf{\hat{f}}^h=\boldsymbol{\Psi}^T \mathbf{y}, \nonumber
\end{equation}
where $\mathbf{A}:=(\boldsymbol{\Psi}^T\boldsymbol{\Psi}+\lambda\mathbf{R})$. Then, the explicit form of $\mathbf{\hat{f}}^h$ is given by:
\begin{equation}
\label{equ::fin_system}
\mathbf{\hat{f}}^h = \mathbf{A}^{-1}\boldsymbol{\Psi}^T  \mathbf{y} = (\boldsymbol{\Psi}^T\boldsymbol{\Psi}+\lambda\mathbf{R})^{-1}\boldsymbol{\Psi}^T \mathbf{y}
\end{equation}
We see from (\ref{equ::fin_system}) that the estimator $\hat{\mathbf{f}}^h$ has the typical form of a penalized least-square estimator. Since $\mathbf{\hat{y}}= \boldsymbol{\Psi}\mathbf{\hat{f}}^h,$ we finally get the evaluation of the function $\hat{f}^h$ in the points $\{\mathbf{p}_1,\dots, \mathbf{p}_N\}$ as:
\begin{equation}
\label{equ::y_hat}
\mathbf{\hat{y}}= \boldsymbol{\Psi}(\boldsymbol{\Psi}^T\boldsymbol{\Psi}+\lambda\mathbf{R})^{-1}\boldsymbol{\Psi}^T \mathbf{y}.
\end{equation}
The smoothing matrix, that maps the observed data values $\mathbf{y}$ in the fitted data values $\hat{\mathbf{y}}$, is given by:
\begin{equation}
\label{equ:hat_matrix}
\mathbf{S}= \boldsymbol{\Psi}(\boldsymbol{\Psi}^T\boldsymbol{\Psi}+\lambda\mathbf{R})^{-1}\boldsymbol{\Psi}^T.
\end{equation}
The trace of $\mathbf{S}$ is a  measure of the equivalent degrees of freedom of the estimator (see \citet{Buja89}). If $\lambda=0$, the number of degrees of freedom is equivalent to the number of basis functions $N^h$. 
\dblue{
However, the two notions differ for $\lambda>0$. While different definitions of equivalent degrees of freedom can be considered, these can be assumed as a consistent measure  of the number of degrees of freedom that takes into account the harmonic penalization. In this respect, if the smoothing parameter $\lambda$ is strictly positive, the number of equivalent degrees of freedom is smaller than the number of basis functions $N^h$ used in $\mathcal{F}^h$.}

\dblue{We use IGA to solve the minimization problem (\ref{equ::min_functional}), for which we define:
$$\mathcal{F}^h = \text{span}\left\{R_{i, p}\circ X^{-1}(\xi, \eta),\quad i=1,\dots, N^h\right\},$$
where $R_{i, p}$ are the NURBS basis functions used to build $\Sigma$,}  eventually after the application of some $h-$, $p-$ or $k-$refinement procedure, as described in Section \ref{sec::NURBS}. $N^h$ is the number of basis functions, which is the dimension of $\mathcal{F}^h$. NURBS allow to define basis functions which are globally $C^1-$continuous on $\Sigma$. As consequence, one can approximate problem (\ref{equ::min_functional}) with the standard Galerkin method, since $\mathcal{F}^h\subset H^2(\Sigma)$; see \citet{Bartezzaghi15} and \citet{Dede2}. In this manner, we obtain a method, which we name IGS, allowing to perform smoothing on surfaces by means of NURBS-based IGA. This also allows encapsulating the original description of the geometry of the surface in the analysis.

The smoothing parameter $\lambda$ may be chosen by minimization of a generalized cross-validation criterion
(GCV), defined as:
\begin{equation}
\label{eq::def_GCV}
\text{GCV}(\lambda)= \frac{N}{\left[N- \trace\left(\mathbf{S}(\lambda)\right)\right]^2}\ \|\mathbf{\hat{y}}(\lambda)-\mathbf{y}\|^2;
\end{equation}
see \citet{Craven78}. Here, we use $\trace\left(\mathbf{S}(\lambda)\right)$ as a measure of the equivalent degrees of freedom (EDF) of the model \citep{Buja89}. In order to solve the optimization problem corresponding to GCV minimization, we use a BFGS quasi-Newton method (see, e.g., \citet{Nocedal99}) with a sufficiently small tolerance. \dblue{One can observe that the computation of the GCV criterion involves the inversion of the matrix $\boldsymbol{\Psi}^T\boldsymbol{\Psi}+\lambda\mathbf{R}$ of size $N^h \times N^h$. In our implementation we use a direct method to compute the matrix $\mathbf{S}$ because of the moderate size of this matrix. Methods based on matrix decompositions can improve both the efficiency and the stability of the optimization procedure used for the GCV criterion, \citep[see, e.g.,][pp. 178--181]{Wood06}.} Other methods for the choice of the smoothing parameter are also available, e.g. the restricted maximum likelihood estimation \citep{Wood11}.

\dblue{
We remark that our model only considers a deterministic location of the measurement points, according to (\ref{equ::math_mod}). While it is easy to account for random location of points at the implementation level of the IGS method, the uncertainty quantification in this setting is not straightforward.}
\subsection{Distributional properties and quantification of uncertainty}
\dblue{Here, we denote by $\mathbb{E}[\mathbf{y}]$ and $ \text{var}(\mathbf{y})$ the expectation and the variance of $\mathbf{y}$ respectively. Moreover, let $\sigma^2$ be the constant variance of the noise introduced in (\ref{equ::math_mod}).}
For a given smoothing parameter $\lambda$, the estimate $\mathbf{\hat{f}}^h$ is a linear transformation of the observations $\mathbf{y}$, as shown in (\ref{equ::fin_system}). Moreover, we have $\mathbb{E}[\mathbf{y}] = \mathbf{f}_N = (f(\mathbf{p}_1),\dots, f(\mathbf{p}_n))^T$ and $\text{var}(\mathbf{y}) = \sigma^2\mathbf{I}$. Then, from (\ref{equ::fin_system}), we get:
$$\mathbb{E}[\mathbf{\hat{f}}^h]=  \mathbb{E}[(\boldsymbol{\Psi}^T\boldsymbol{\Psi}+\lambda\mathbf{R})^{-1}\boldsymbol{\Psi}^T \mathbf{y}] = \mathbf{A}^{-1}\boldsymbol{\Psi}^T \mathbf{f}_N.$$
Similarly, we can directly express the variance as:
\begin{equation}
\label{equ::var_fh}
 \text{var}(\mathbf{\hat{f}}^h) =  \text{var}(\mathbf{A}^{-1}\boldsymbol{\Psi}^T \mathbf{y}) 
= \sigma^2 \mathbf{A}^{-1}\boldsymbol{\Psi}^T\boldsymbol{\Psi}\mathbf{A}^{-1},
\end{equation}
where we used the fact that the matrix $\mathbf{A}$ is symmetric and hence $\mathbf{A}^{-T} = \mathbf{A}^{-1}$.  In particular, under the assumption of normality of the errors, we have:
$$\mathbf{\hat{f}}^h \sim \mathcal{N}(\mathbf{A}^{-1}\boldsymbol{\Psi}^T \mathbf{f}_N,\sigma^2 \mathbf{A}^{-1}\boldsymbol{\Psi}^T\boldsymbol{\Psi}\mathbf{A}^{-1} ).$$
Then, we can also express the expectation and the variance of the fitted values $\mathbf{\hat{y}}$ explicitly as:
$$\mathbb{E}\left[\mathbf{\hat{y}}\right] = \mathbf{S}\ \mathbf{f}_N,$$
and:
$$\text{var}(\mathbf{\hat{y}})=\sigma^2 \mathbf{S} \mathbf{S}^T = \sigma^2\ \mathbf{S}^2,$$
respectively, since the smoothing matrix $\mathbf{S}$ is symmetric. In practice, the error variance $\sigma^2$ must be estimated from the data. Following \citet{Hastie90} we can estimate $\sigma^2$ by:
\begin{equation}
\label{equ:est_var}
\hat{\sigma}^2 = \frac{\|\mathbf{\hat{y}}- \mathbf{y}\|^2}{N  - \trace{(\mathbf{S})}}.
\end{equation}
Given an additional point $\mathbf{p}_{N+1}$ on $\Sigma$, the predicted value of the function $f$ is given by:
$$\hat{f}^h(\mathbf{p}_{N+1}) = \sum_{i=1}^{N^h} \hat{f}_i\ \psi_i(\mathbf{p}_{N+1}) = \boldsymbol{\psi}(\mathbf{p}_{N+1})^T \mathbf{\hat{f}}^h,$$
while \dblue{its} variance is given by
$$ \text{var}(\hat{f}^h(\mathbf{p}_{N+1})) = \boldsymbol{\psi}(\mathbf{p}_{N+1})^T \text{var}(\mathbf{\hat{f}}^h) \boldsymbol{\psi}(\mathbf{p}_{N+1}) = \sigma^2 \boldsymbol{\psi}(\mathbf{p}_{N+1})^T  \mathbf{A}^{-1}\boldsymbol{\Psi}^T\boldsymbol{\Psi}\mathbf{A}^{-1}\boldsymbol{\psi}(\mathbf{p}_{N+1}).$$
An estimate of $\text{var}(\hat{f}^h(\mathbf{p}_{N+1}))$, say $\widehat{\text{var}}(\hat{f}^h(\mathbf{p}_{N+1})$, reads:
$$ \widehat{\text{var}}(\hat{f}^h(\mathbf{p}_{N+1}) = \hat{\sigma}^2 \boldsymbol{\psi}(\mathbf{p}_{N+1})^T  \mathbf{A}^{-1}\boldsymbol{\Psi}^T\boldsymbol{\Psi}\mathbf{A}^{-1}\boldsymbol{\psi}(\mathbf{p}_{N+1}).$$

These results fully characterize the estimates in the case of Gaussian noise. In such case, one can derive confidence bands on the estimated functions
and thus quantifying the uncertainty of the estimations of any predicted value. If the Gaussian assumption does not hold, this confidence interval should be used with caution and the confidence level is only approximated. More generally, the quantification of uncertainty has been widely studied in the context of generalized additive models (see, e.g., \citet{Wood06} and references therein). A Bayesian approach to uncertainty quantification for this class of models is also possible; see \citet{Marra12}.
\section{Numerical simulations}
\label{sec:num}
In order to assess the IGS methodology, we consider two simulations on surfaces represented by NURBS for which the function $f$ to be estimated is given a priori.
We aim at showing different properties of IGS in different settings.  In the first simulation, the configuration is such that any method used for planar smoothing such as Thin Plate Splines (TPS) \citep{Duchon77, Wahba90} can be efficiently used and thus compared to IGS. In the second simulation, the setting is such that methods for two dimensional smoothing are not appropriate, while IGS can be straightforwardly used.
\subsection{Simulation 1}
\label{S:num1}
We consider a quarter of cylinder $\Sigma$ defined as $\Sigma=\{x^2+y^2=1,\quad
0<z<2,\quad x>0, \quad y>0\}$.
Using cylindrical coordinates, it is parametrized by the following mapping:
$$ X: \Omega = (0,2)\times\left(0,\frac{\pi}{2}\right)\rightarrow \Sigma, \quad X(s_1,s_2) = (\cos(s_1), \sin(s_1), s_2).$$
This is an \emph{isometric} mapping \citep{Stoker89}. That means that the mapping preserves lengths and angles, and thus the area. In other words, the parametric domain $\Omega$ is not distorted by the mapping. This kind of mapping allows to indifferently work on the parametric domain $\Omega$ or directly on the surface $\Sigma$. Thus, the smoothing can be performed on the parametric domain, which is planar, namely using any traditional method for planar smoothing. In particular, in the following, we shall compare the proposed technique to TPS.

The surface $\Sigma$ is exactly representable using NURBS basis functions of degree 2 or higher which are at least globally $C^1-$continuous. We are interested in recovering the function:
$$f(x,y,z)= \sin\left(\frac{5 \pi}{2}\left[ xy^2 - y \left( \frac{z}{2} - 1 \right)^2 + x^2 \left( \frac{z}{2} - 1 \right)\right]  + \frac{\pi}{3}\right),$$
from noisy and discrete observations. The quarter of cylinder $\Sigma$ and the function $f$ are shown in Figure~\ref{fig:num.1.ex}. The function $f$ is evaluated in $N=100$ points $\left\{\mathbf{p}_i\right\}_{i=1}^N$ \dblue{located on a 10$\times$10 grid of equally spaced points in the parametric domain} and is affected by independent Gaussian observational errors $\varepsilon_i$ of zero mean and standard deviation $\sigma = 0.125$. We generate the data $y_i=f(\mathbf p_i)+\epsilon_i$ for $i=1,\ldots,N$ and the simulation is repeated $M=100$ times. \dblue{We use random scalars drawn from the standard normal distribution to generate the noise (specifically, we used the MATLAB function \texttt{randn})}. The dimension of the IGA space $\mathcal F^h$ varies between $N^h=49$ and $N^h=121$ and we use different NURBS basis functions, namely globally $C^1-$, $C^2-$, and $C^3-$ continuous of degrees $p=2,3$, and $4$ respectively, obtained from $k-$refinement. The smoothing parameter $\lambda$ is chosen at each simulation repetition and for each basis setting considered, by minimizing  GCV criterion in (\ref{eq::def_GCV}).
\begin{figure}[t!]
\centering
\includegraphics[trim=100 0 100 30,scale=.25]{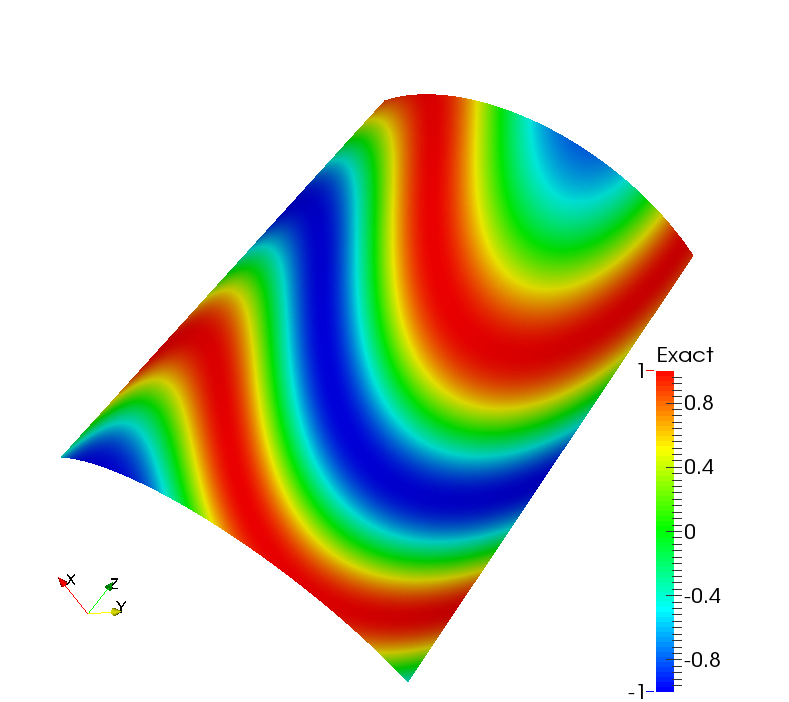}
\caption{\label{fig:num.1.ex}Simulation~$1$. Surface $\Sigma$ (quarter of cylinder) and exact function $f$.}
\end{figure}

\begin{figure}[htbpt]
\centering
\begin{subfigure}{.49\linewidth}
\includegraphics[scale=0.23]{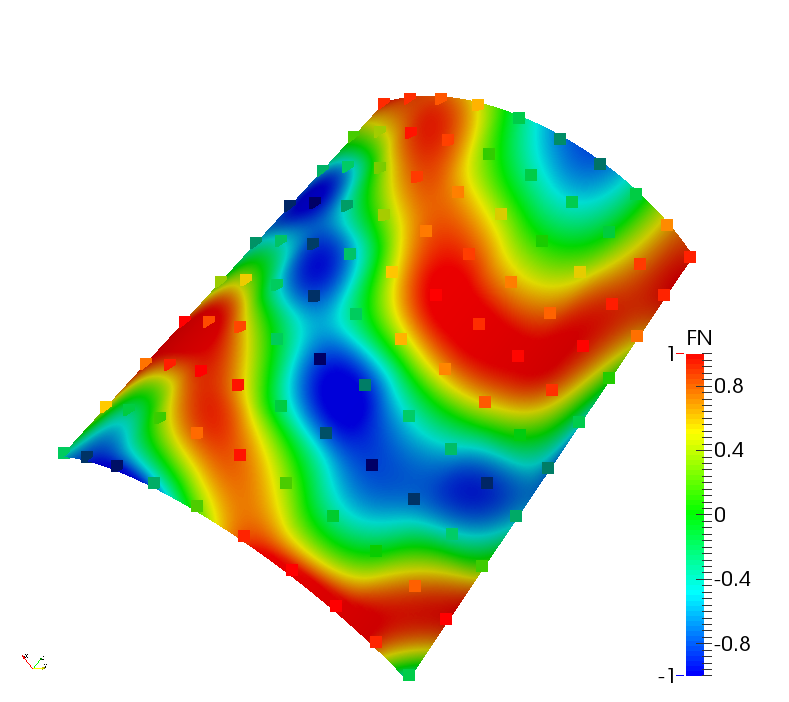} 
\subcaption{$N^h = 49$}
\end{subfigure}
\begin{subfigure}{.49\linewidth}
\includegraphics[scale=0.23]{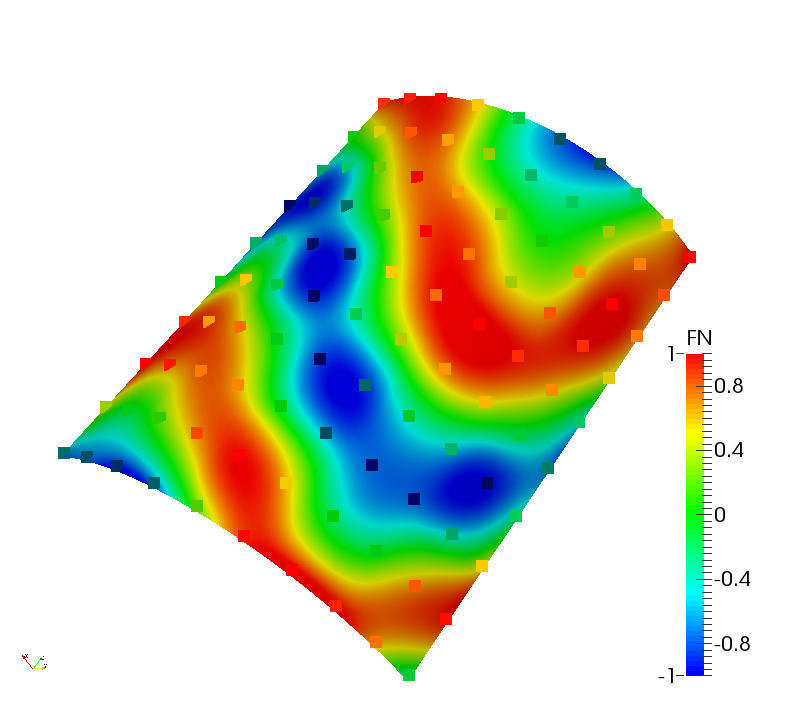}
\subcaption{$N^h = 49$}
\end{subfigure}
\begin{subfigure}{.49\linewidth}
\includegraphics[scale=0.23]{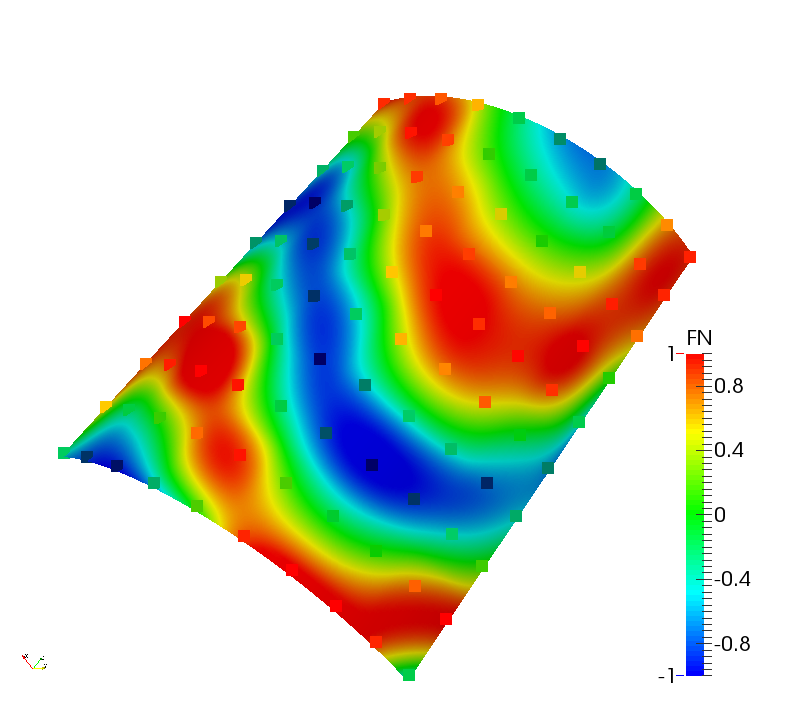}
\subcaption{$N^h = 64$} 
\end{subfigure}
\begin{subfigure}{.49\linewidth}
\includegraphics[scale=0.23]{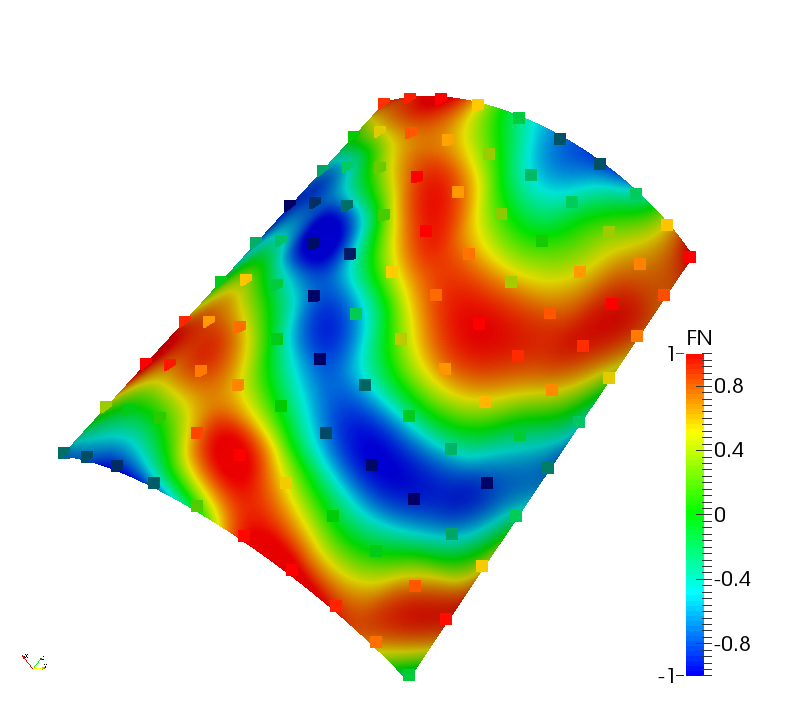}
\subcaption{$N^h = 64$}
\end{subfigure}
\begin{subfigure}{.49\linewidth}
\includegraphics[scale=0.23]{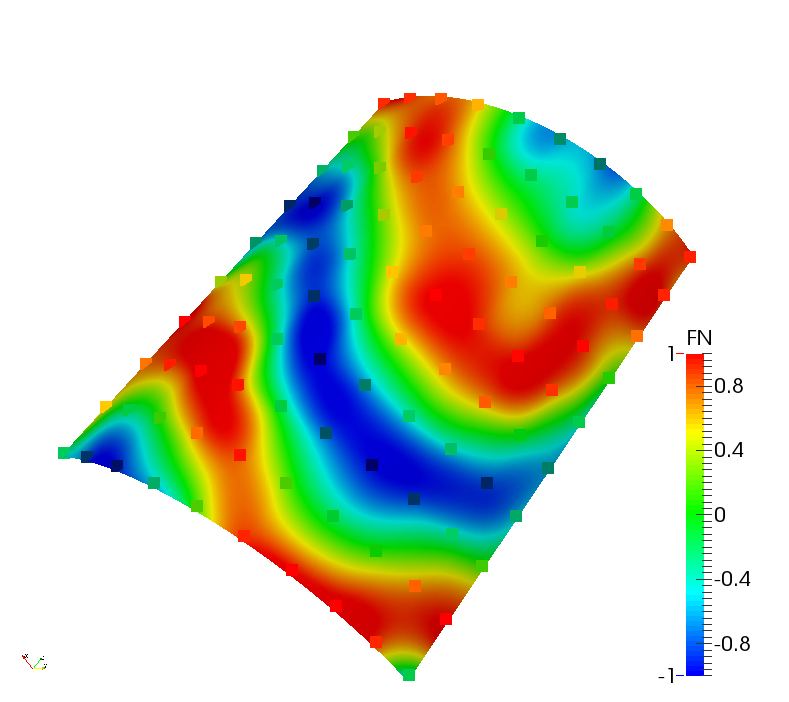}
\subcaption{$N^h = 121$}
\end{subfigure}
\begin{subfigure}{.49\linewidth}
\includegraphics[scale=0.23]{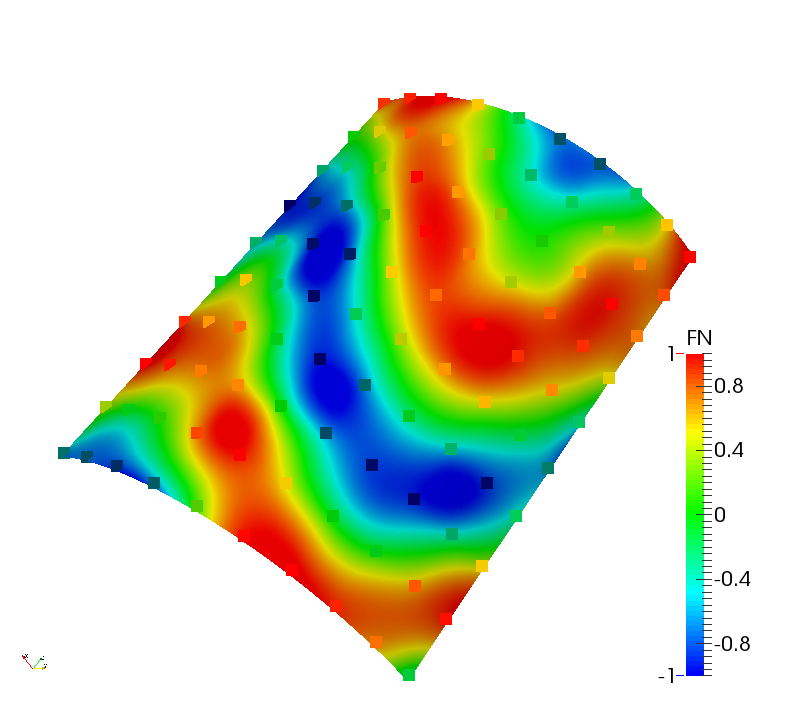}
\subcaption{$N^h = 121$}
\end{subfigure}
\caption{\label{fig:num.1.fn}Simulation $1$. Estimated functions $\hat{f}^h$ in the first (left) and second (right) simulation repetition, out of $M=100$ repetitions, using  $N^h = 49$ (top), $N^h = 64$ (middle) and $N^h = 121$ (bottom) number of basis functions, globally $C^1-$continuous of degree $p=2$. The corresponding measured values $\left\{y_i\right\}_{i=1}^N$ are displayed on the same color scale as the true function and estimates.}
\end{figure}
\begin{figure}[t]
\centering
\begin{subfigure}{.49\linewidth}
\includegraphics[scale=0.24]{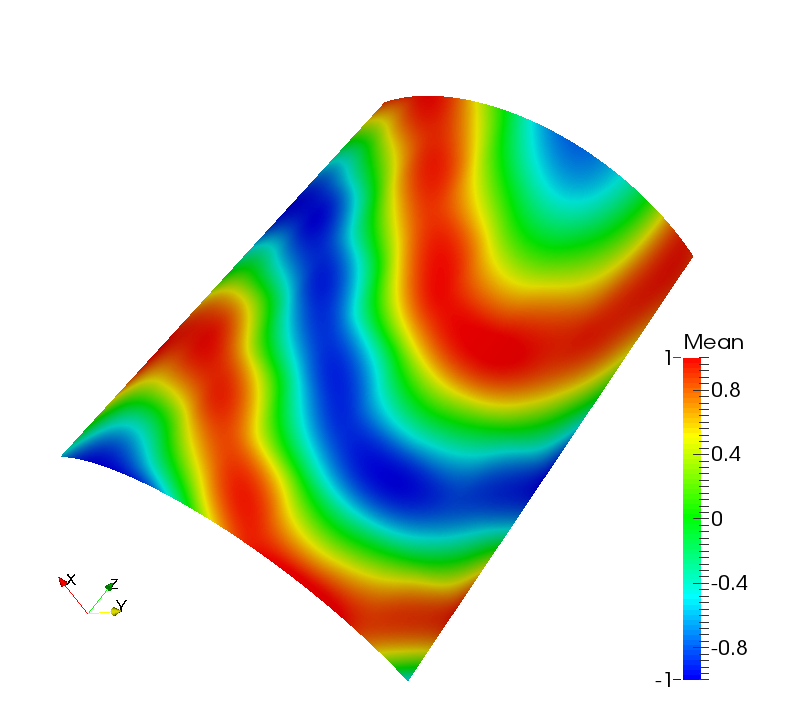} 
\end{subfigure}
\begin{subfigure}{.49\linewidth}
\includegraphics[scale=0.24]{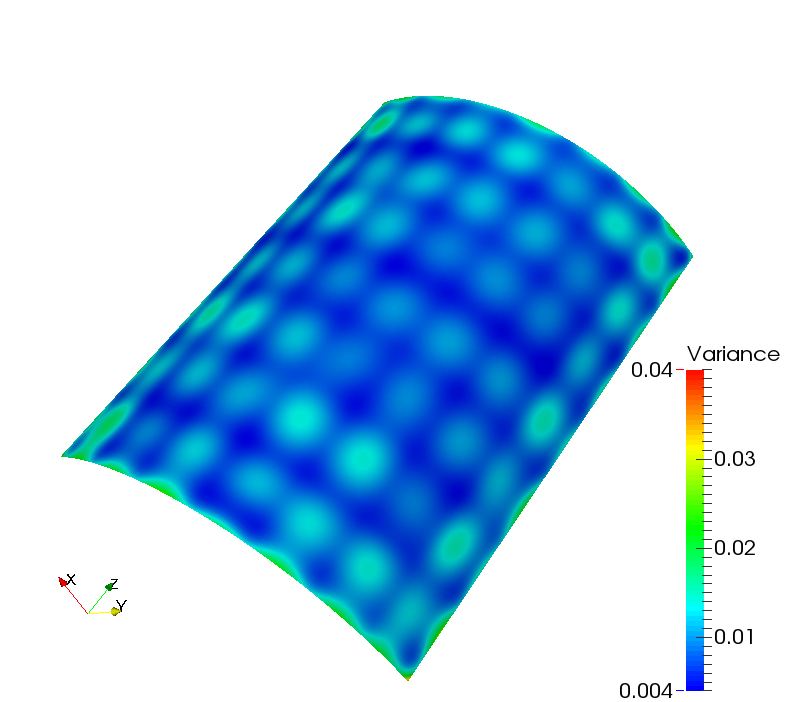}
\end{subfigure}
\caption{\label{fig:num.1.n100}Simulation~$1$. Empirical mean function $\overline{\hat f^h}$ (left) and empirical variance of the function $\hat{f}^h$~(right) over the M=100 simulation repetitions, for NURBS basis functions of degree $p=2$, globally $C^1-$ continuous and of dimension $N^h=81$. }
\end{figure}
The first two estimated functions $\hat{f}^h$ over the $M$ ones are displayed in Figure~\ref{fig:num.1.fn}. We observe that these are qualitatively good estimates of $f$. In general, the number of basis functions $N^h$ must be chosen carefully. Indeed, this choice should depend both on the complexity of the function to be estimated $f$ and of the number of data points $N$ available. When the number of basis functions $N^h$ is small, IGS is not able to capture the behaviour of the function $f$, as it would be the case with any other smoother. On the contrary, when the number of basis functions $N^h$ is high, we see that there is a larger variability in the estimated functions $\hat{f}^h$. Indeed, when the number of basis used is too high, GCV criterion can lead to overfitting, that is the estimated function incorporates noise. However, we see in Figure~\ref{fig:num.1.fn} that the estimates are not very sensitive to this choice. Finally, we remark that the minimum number of basis functions is basically dictated by the number of functions used to represent the surface with NURBS.

We notice, following Section \ref{sec::math_model}, that we used natural boundary conditions, for which we have not formally proved the well-posedness of the problem. We report in Table \ref{tab:cond_numb} the mean condition number $K_\infty$ for the matrix $A$ of (\ref{equ::fin_system}), with $\lambda$ chosen with the GCV criterion and for $p=2$ only, since results for $p=3$ and $p=4$ are similar. As a matter of fact, the system (\ref{equ::fin_system}) results to be well-conditioned in all our numerical experiments.
\begin{table}[t]
\centering
\begin{tabular}{l|c|c|c|c|c|c|c|c|c|c}
$N^h$ & $25$  & $36$  & $49$  & $64$  & $81$  & $100$  & $121$  & $144$  & $169$  & $196$ \\
\hline
$K_\infty(\mathbf{A})$ &  $46.1$ & $47.3$ & $48.2$ & $46.5$ & $43.9$ &  $80.9$ &  $210$ & $431$ & $355$ & $307$
\end{tabular}
\caption{\label{tab:cond_numb}Simulation 1. Mean condition number $K_\infty(\mathbf{A})$ for the matrix $\mathbf{A}$, with smoothing parameter chosen with GCV, $p=2$ and for different number of basis functions $N^h$.}
\end{table}
In order to better assess the IGS method, we use the empirical mean function $\overline{\hat f^h}=\frac{1}{M} \sum_{i=1}^M \hat{f_i^h}(\mathbf{x})$ and the associated empirical variance function $\frac{1}{M} \sum_{i=1}^M (\hat{f_i^h}(\mathbf{x}) -\overline{\hat {f^h}}(\mathbf{x}))^2,$ where $\hat{f}_i^h$ denotes the $i$-th estimated function. These are both shown in Figure~\ref{fig:num.1.n100}, in the case of globally $C^1-$continuous basis functions of degree $p=2$ and with $N^h=81$ basis functions.  The estimates appears to have a negligible bias and a relatively small variance.

\begin{figure}[t!]
\centering
\begin{subfigure}{0.49\linewidth}
\includegraphics[scale= 0.53]{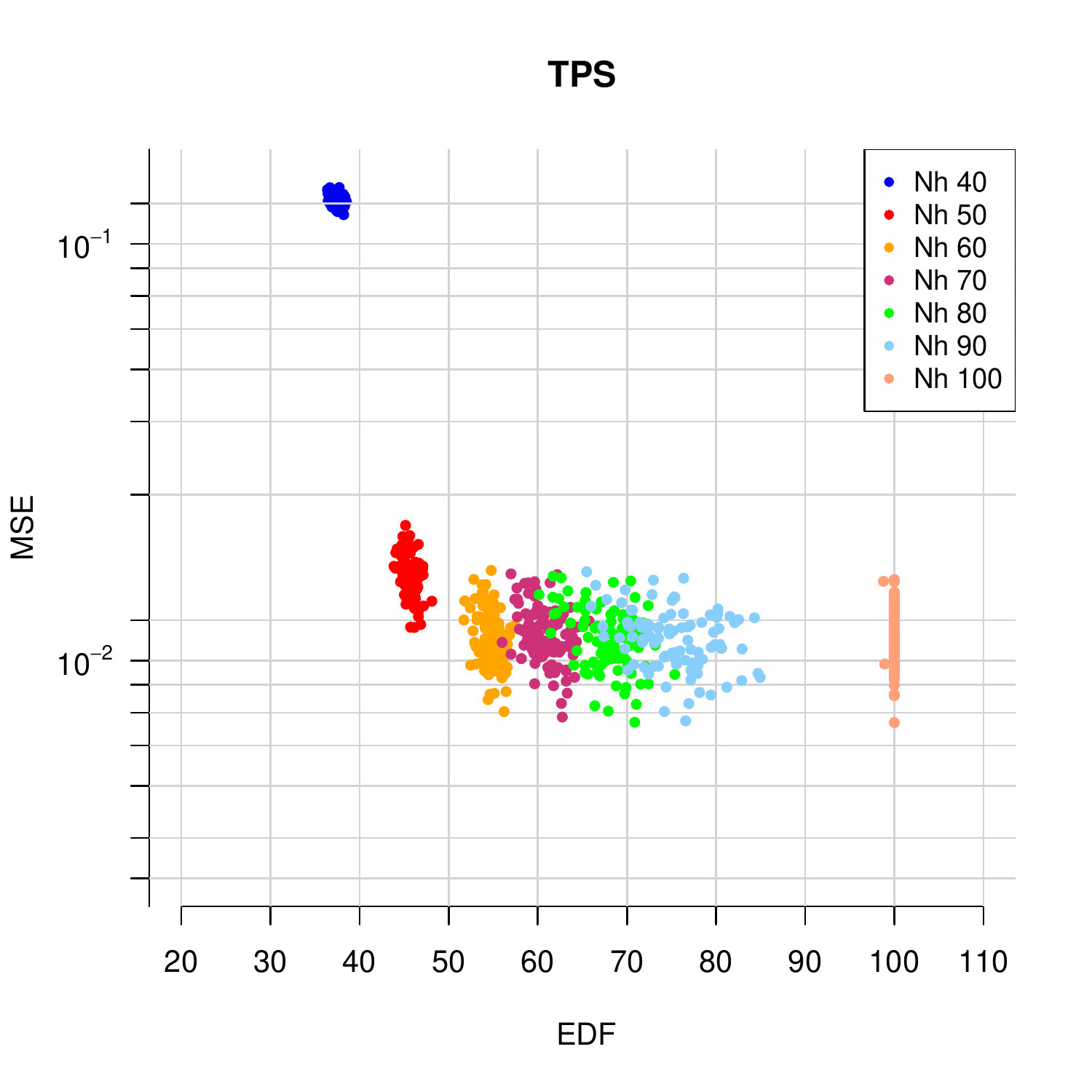}
\end{subfigure}
\begin{subfigure}{0.49\linewidth}
\includegraphics[scale= 0.53]{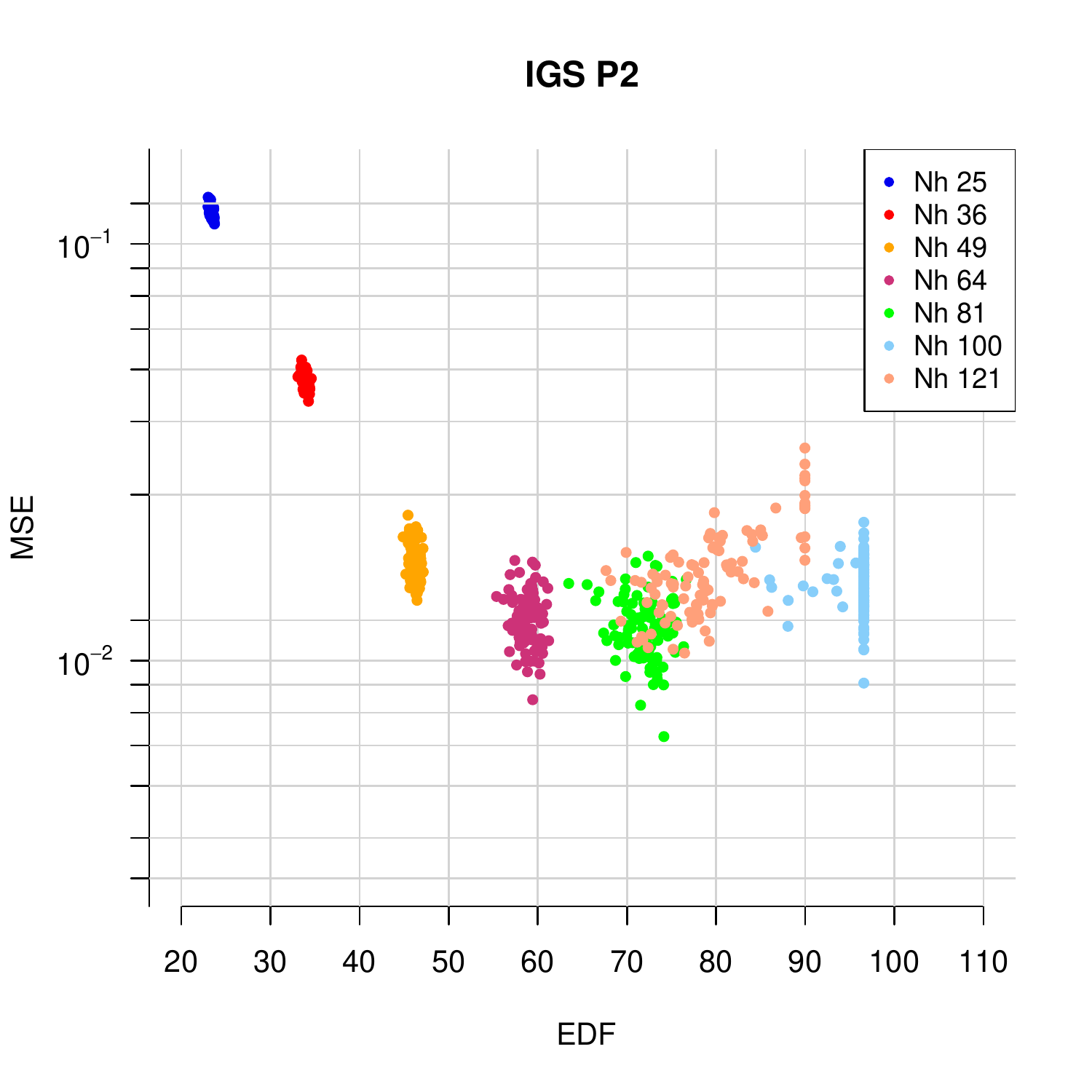}
\end{subfigure}
\begin{subfigure}{0.49\linewidth}
\includegraphics[scale= 0.53]{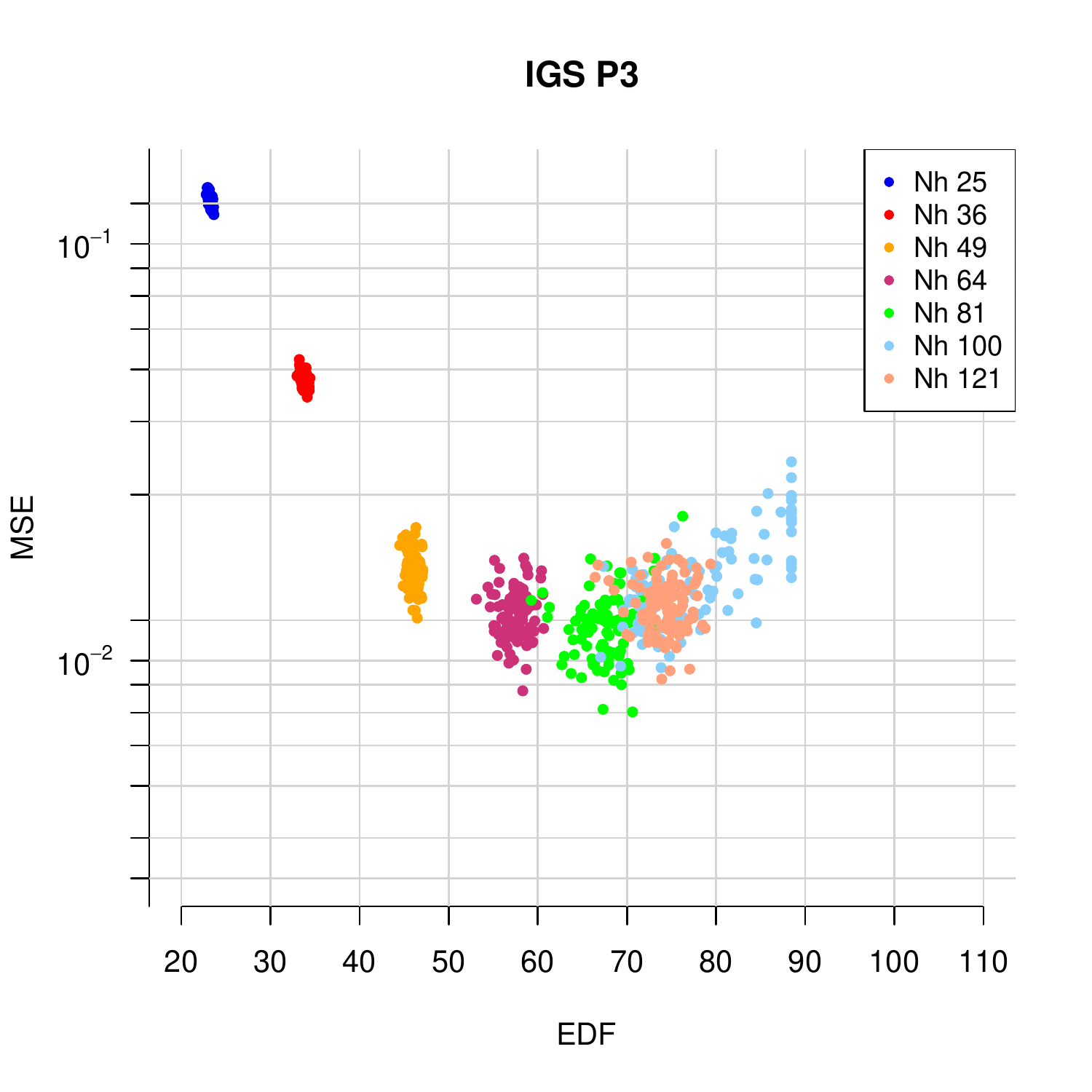}
\end{subfigure}
\begin{subfigure}{0.49\linewidth}
\includegraphics[scale= 0.53]{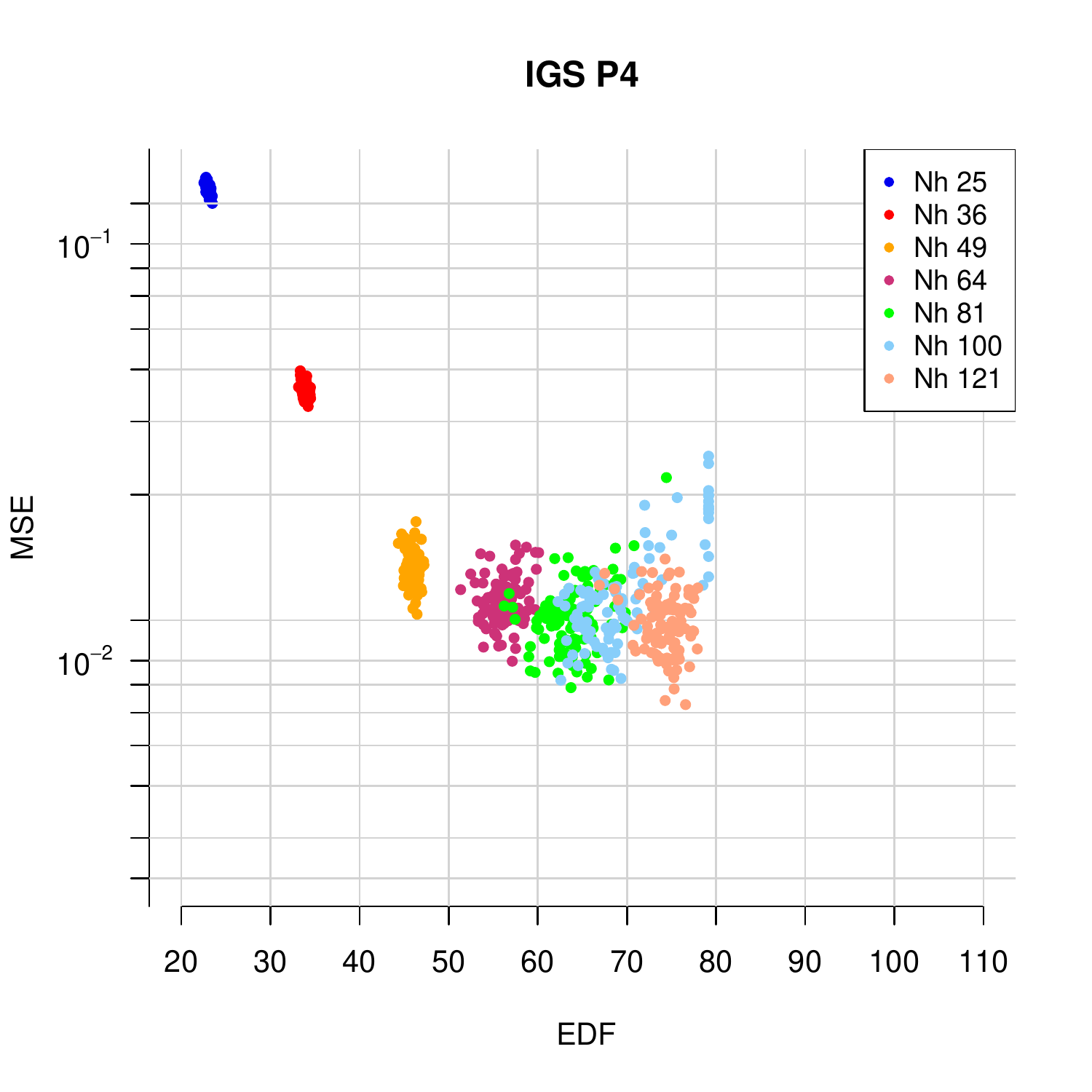}
\end{subfigure}
\caption{\label{fig:comp_methods_sig0125}Simulation~$1$. Comparisons of the MSE for IGS and TPS (top-left), in terms of the equivalent degrees of freedom (EDF). IGS uses globally $C^1-$, $C^2-$ and $C^3-$continuous basis functions of degrees $p=2$ (top-right), $3$ (bottom-left) and $4$ (bottom-right), respectively.}
\end{figure}

We compare our methodology with a widely used smoothing technique, TPS, using cylindrical coordinates. This method is implemented for instance in the R package \texttt{mgcv} \citep{mgcv}. We use different number of basis functions $N^h$ for a comparison, selecting the smoothing parameter at each simulation repetition and for each number of basis considered via GCV. As a criterion for the comparison, we use the mean squared error (MSE) of the estimator, computed as:
$$\text{MSE}= \frac{1}{L}\sum_{j=1}^L \left(\hat{f}^h(\mathbf{l}_j)- f(\mathbf{l}_j)\right)^2, $$
where $\mathbf{l}_1, \dots,\mathbf{l}_L$ is a lattice of $150\times 200$ evaluation points on $\Sigma$. We compute the MSE for $N^h=49,64,81, 100$, and $121$ and degree $p=2,3,$ and $4$ for IGS and for $N^h= 40,50,\dots,100$ for TPS. The comparisons of MSE in terms of the equivalent degree of freedom is shown in Figure~\ref{fig:comp_methods_sig0125}. We also compare in Figure  \ref{fig:comp_methods_best} the best setting of each methodology and for each number of basis functions, that is the lowest median MSE. Finally, Figures \ref{fig:comp_methods_sig0125} and \ref{fig:comp_methods_best} show that IGS is comparable to TPS in terms of performance, in a setting where the latter technique may be applied. 
\dblue{Figure \ref{fig:comp_methods_sig0125} illustrates a key feature of the IGS method and more generally of all smoothing techniques. Specifically, one can notice that increasing the number of basis functions does not improve the estimated function $\hat{f}^h$. Indeed, although we increase the number of basis functions to build $\hat{f}^h$, the number of measurement points $N$ remains the same, i.e. $\hat{f}^h$ is still built from the same set of data, but only through a richer finite dimensional space $\mathcal{F}^h$. 
Therefore, the convergence of $\hat{f}^h$ to $f$ should be simultaneously regarded through the number of data $N$ and the quality of the NURBS space $\mathcal{F}^h$. We highlight in Figure  \ref{fig:comp_methods_sig0125} that the same happens with TPS-based smoothing. 
A key role in the convergence of the method may be played by the NURBS space $\mathcal{F}^h$ when the number of measurement points is clustered in a region of $\Sigma$, for which mesh refinement techniques can be used.} 

\begin{figure}[t!]
\centering
\includegraphics[scale= 0.5]{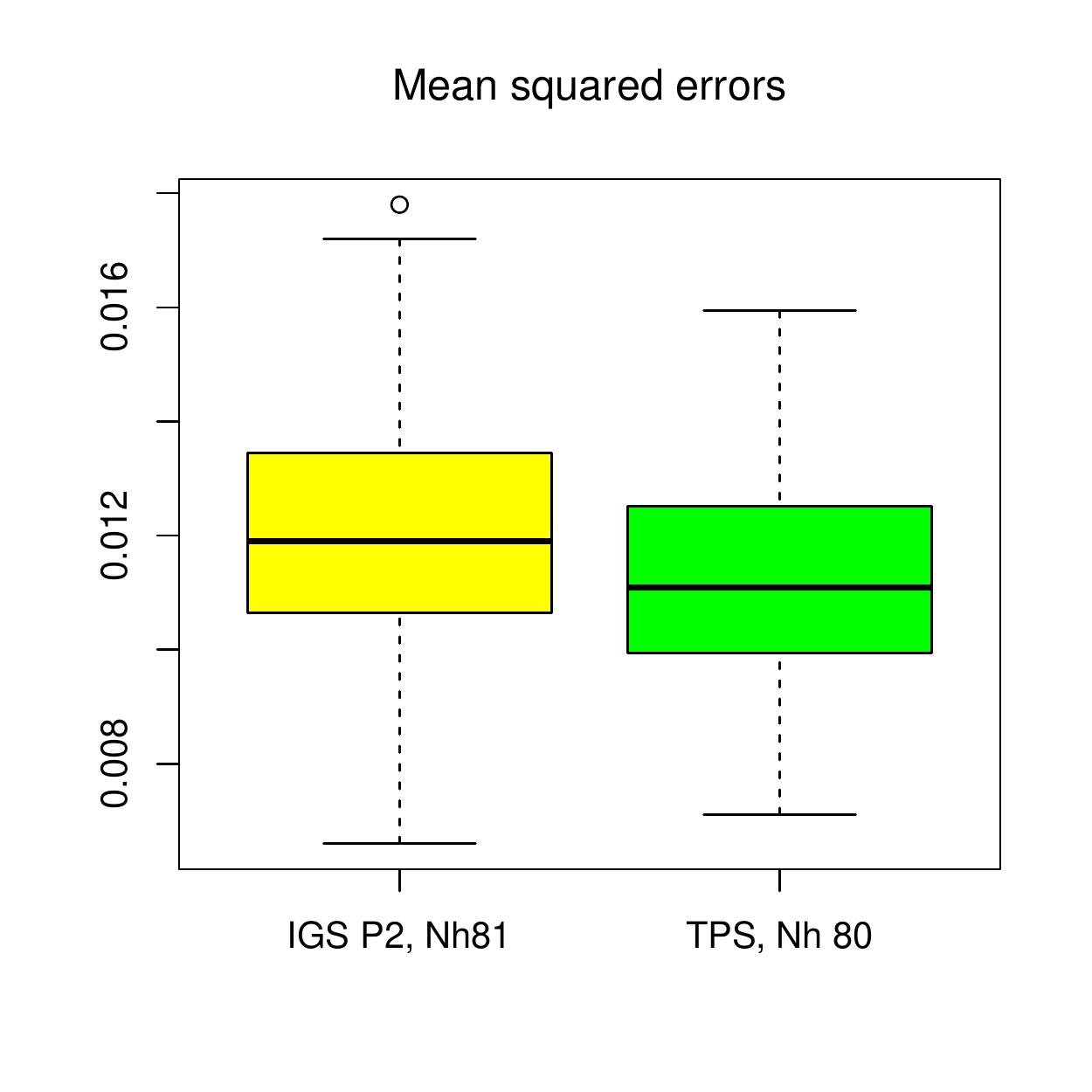} 
\caption{\label{fig:comp_methods_best}Simulation~$1$. Comparisons of the MSE for the best setting of IGS and TPS.}
\end{figure}

\subsection{Simulation 2}
\label{S:num.2}
We consider the surface $\Sigma$ reported in Figure~\ref{fig:num.2.ex}  which is represented in terms of NURBS basis functions of degree $p=2$ starting from the knot vector $\varXi=\{0,0,0,1,1,1\}$ along both the parametric directions and control points $\mathbf P_1=(1,0,0)^T$, $\mathbf P_2=(0,0.75,0)^T$, $\mathbf P_3=(0,1,0)^T$, $\mathbf P_4=(1,0,1)^T$, $\mathbf P_5=(1,1,0.5)^T$, $\mathbf P_6=(0,1,1)^T$, $\mathbf P_7=(0.25,0,1)^T$, $\mathbf P_8=(0.25,0.25,0.25)^T$, and $\mathbf P_9=(0.25,0,1)^T$, and with the corresponding weights vector $\mathbf w = (1,1/\sqrt{2},1,1/\sqrt{2},1/2,1/\sqrt{2},1,1/\sqrt{2},1)^T$. As reported in Figure~\ref{fig:num.2.ex}, we consider the exact function to be estimated:
$$f(x,y,z)=\left( 2 \, \left( x \sin( \pi  y )  + \cos( \pi x)   y \right) - 1 \right)  \,  \cos\left( \frac{5}{4} \pi z \right),$$ which is evaluated in $N=100$ points $\left\{\mathbf p_i\right\}_{i=1}^N$ located on $\Sigma$; \dblue{more specifically, they are located on a 10$\times$10 grid of equally spaced points in the parametric domain}. The $N$ observations are affected by independent Gaussian observational errors of zero mean and standard deviation $\sigma=0.125$. The sampling of the data $y_i=f(\mathbf p_i)+\epsilon_i$, $i=1,\ldots,N$, is repeated $M=50$ times. The sampling locations $\left\{\mathbf p_i\right\}_{i=1}^N$ remain the same for all the repetitions. We consider three NURBS spaces $\mathcal F^h$ of dimensions $N^h=100$, of globally $C^1-$, $C^2-$ and $C^3-$ continuous NURBS basis functions of degrees $p=2,3$, and $4$, respectively. These are obtained by $k-$refinement of the original NURBS basis.

\begin{figure}[tp!]
\centering
\includegraphics[scale=.17]{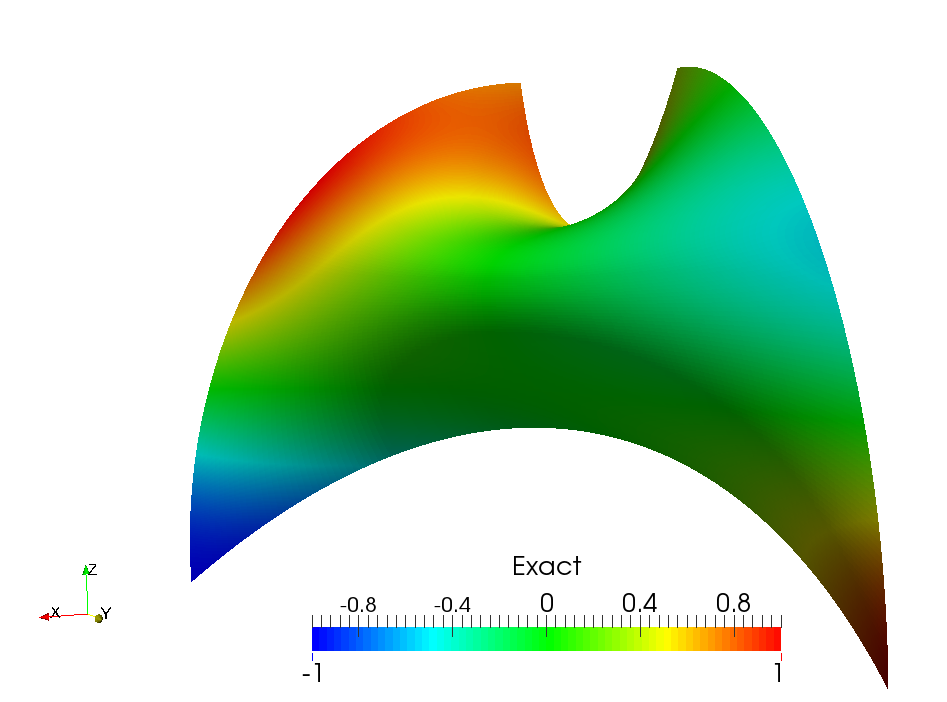}
\caption{Simulation~$2$. Surface $\Sigma$ and exact function $f$}
\label{fig:num.2.ex}
\end{figure}

\begin{figure}[t!] 
\centering
\begin{subfigure}[b]{.5\linewidth}
\includegraphics[scale=0.23]{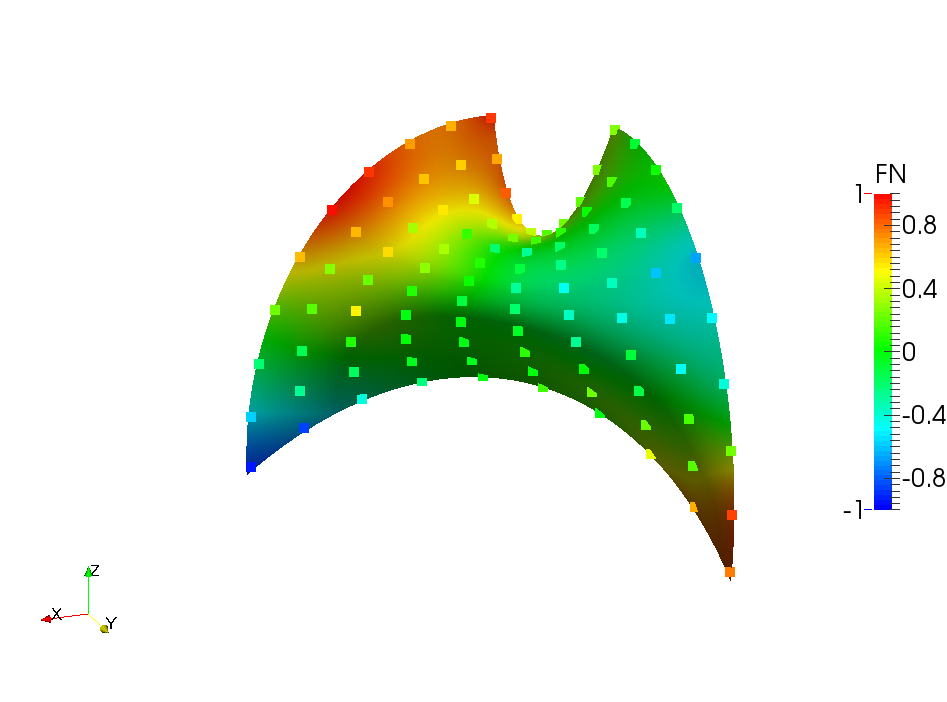} 
\end{subfigure}%
\begin{subfigure}[b]{.5\linewidth}
\includegraphics[scale=0.23]{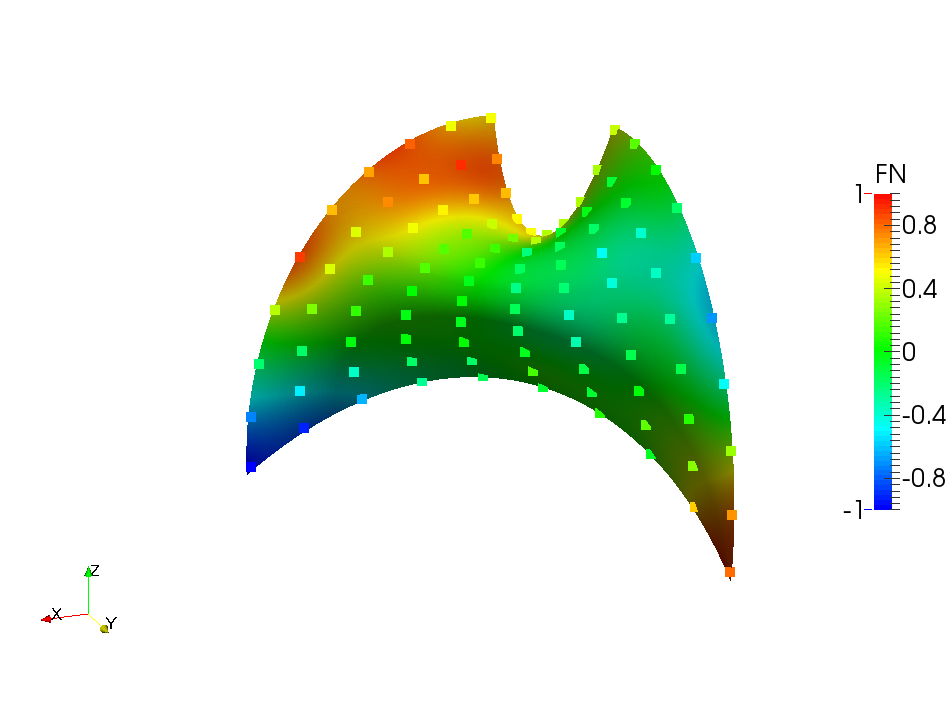}
\end{subfigure}%
\caption{\label{fig:num.2.fn}Simulation~$2$. Functions $\hat f^h$ estimated for the last two repetitions, for $N=100$ points; positions of the points $\left\{\mathbf p_i\right\}_{i=1}^N$ and the corresponding measured values $\left\{y_i\right\}_{i=1}^N$ for NURBS basis functions of degrees $p=3$ and globally $C^2-$continuous on $\Sigma$.}
\end{figure}
\begin{figure}[htbp]
\centering
\begin{subfigure}[b]{.5\linewidth}
\includegraphics[scale=.17]{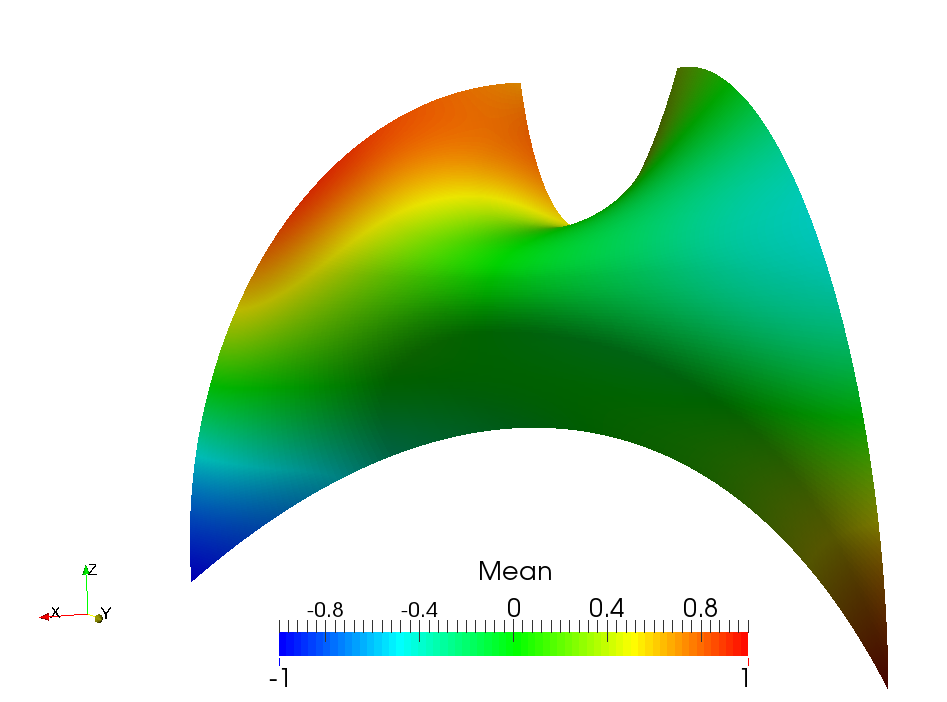}
\end{subfigure}%
\begin{subfigure}[b]{.5\linewidth}
\includegraphics[trim=10 0 2cm 0, clip,scale=.17]{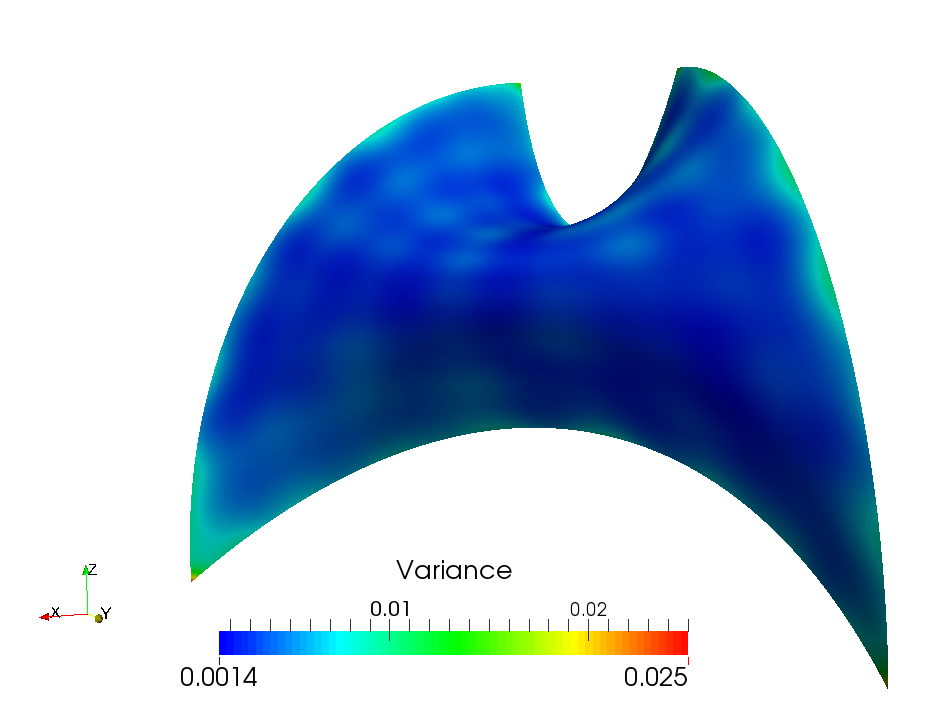}
\end{subfigure}
\caption{\label{fig:num.2.n100}Simulation~$2$. Empirical mean function $\overline{\hat f^h}$ (left) and empirical variance of the function $\hat{f}^h$~(right) for $N=100$ and $M=50$ obtained with NURBS basis functions of degrees $p=3$ and globally $C^2-$continuous on $\Sigma$.}

\end{figure}
In Figure~\ref{fig:num.2.fn} we report the estimated functions $\hat f^h$ for the last two simulation repetitions,  by considering NURBS spaces of basis functions of degree $p=3$ and globally $C^2-$continuous on $\Sigma$. Figure~\ref{fig:num.2.n100} highlights the empirical mean function $\overline{\hat f^h}$ and the corresponding empirical variance function over the $M=50$ simulation repetitions,  obtained for the same NURBS basis functions. The results obtained for the degrees $p=2$ and $4$ are very similar and are not reported here for the sake of brevity. Our estimation has a negligible bias and a small variance, as shown in Figure~\ref{fig:num.2.n100}. The behaviour of the function $f$ seems to be very well captured by our estimator $\hat{f}^h$, as shown in Figure~\ref{fig:num.2.fn}. In this setting, we increase the regularity of the basis functions without changing the number of basis functions $N^h$. The estimated functions are thus globally $C^{p-1}-$continuous. We then compare the mean integrated squared errors (MISE) of IGS for different regularity and degrees of NURBS basis functions.  The MISE of any estimated  function $\hat{f}$ is defined as:
$$\text{MISE}(\hat{f}) = \frac{1}{\left| \Sigma \right|}\int_\Sigma \left(\hat{f} - f\right)^2 \ d\Sigma. $$
As we observe in Figure \ref{fig:RMSE_degree}, the quality of the estimation is not significantly affected by the regularity of the basis functions.

\begin{figure}[t!]
\centering
\includegraphics[scale= 0.5]{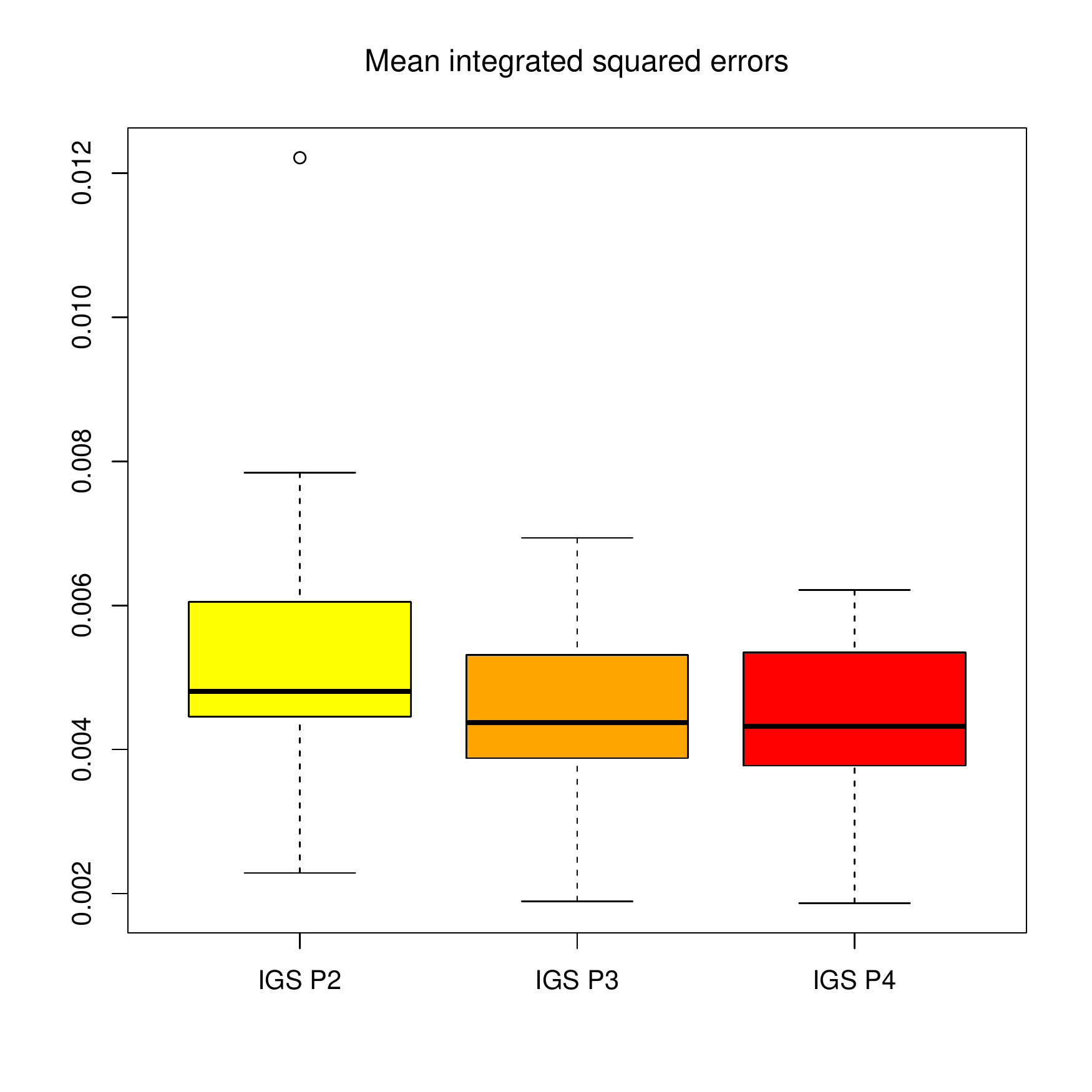}
\caption{Simulation 2. Boxplots of the MISE for IGS with degrees $p=2,3,$ and $4$, globally $C^1-, C^2-$ and $C^3-$continuous basis functions.}
\label{fig:RMSE_degree}
\end{figure}

\section{Estimation of aerodynamic force on the SOAR's winglet}
\label{sec:soar}
We aim at estimating the pressure coefficient field $C_p$ and the corresponding aerodynamic force on the inboard winglet $\Sigma$ of the SOAR shuttle shown in Figure~\ref{fig::shuttle}. The pressure coefficient is a dimensionless field related to the pressure field. It describes the relative pressure over the surface of the winglet and is given by:
$$C_p(\mathbf{x}) = \frac{p(\mathbf{x})- p_\infty}{\frac{\rho_\infty}{2}\  v_\infty^2},$$
where $p(\mathbf{x})$ is the pressure at point $\mathbf{x}$, $v_\infty$ and $p_\infty$ are the far field wind speed and pressure, with $\rho_\infty$  the air density.
The flow regime of the free stream is subsonic, namely the Mach number is $0.7$, for which we can reasonably assume that the pressure coefficient field over the winglet remains sufficiently smooth, since transonic effects are marginal. The force $\mathbf{F}$ acting on the winglet and due to the contribution of the pressure (see, e.g., \citet{Anderson10}) is given by:
$$ \mathbf{F} =\int_\Sigma \left(p- p_\infty\right)\ \mathbf{n}_\Sigma\  d\Sigma = \frac{\rho_\infty}{2}\  v_\infty^2 \int_\Sigma C_p\ \mathbf{n}_\Sigma\  d\Sigma,$$
where $\mathbf{n}_\Sigma$ is the unit normal vector to the surface $\Sigma$. The final quantity of interest of the estimation is the aerodynamic force $\mathbf{F}$.
The pressure coefficient $C_p$ is measured in $N=824$ data points on the surface $\Sigma$ for which the sampled data are depicted in 
Figure~\ref{fig:CP_est_full}, which shows the data after application of an affine transformation\footnote{The displayed data have been modified for copyright reasons with respect to the original ones provided by \textit{S3, Swiss Space Systems Holding SA}.}, that modifies the values of about 3\% at most. These data are derived from Computational Fluid Dynamics (CFD) simulations and represent a preliminary study prior experimental campaign in a wind tunnel. The inboard winglet is represented by a single NURBS patch surface built with $N^h=584$ degrees of freedom. The  geometry is built with a NURBS basis of degree $p=5$ and functions globally $C^2-$continuous. Thus, the minimum number of basis functions representing the function space $\mathcal{F}^h$ for IGS is $N^h=584$.

First, we estimate the pressure coefficient field $\hat{C}^h_p$ with all the available points. We use IGS to estimate the pressure coefficient field and then the aerodynamic force. By using IGS with $N=824$ and $N^h= 584$, for NURBS basis functions of degree $5$ and globally $C^2-$continuous basis functions, we obtain the $\hat{C}^h_p$ field reported in Figure \ref{fig:CP_est_full} (right).

\begin{figure}[t!]
\centering
\begin{subfigure}[b]{.5\linewidth}
\includegraphics[scale=.25]{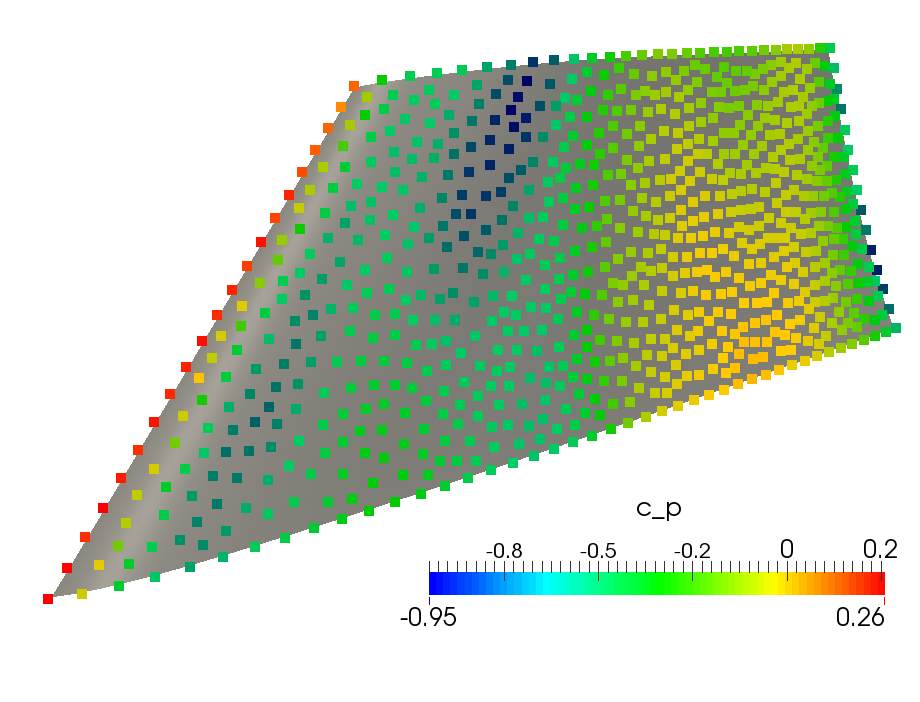}
\end{subfigure}%
\begin{subfigure}[b]{.5\linewidth}
\includegraphics[scale=.25]{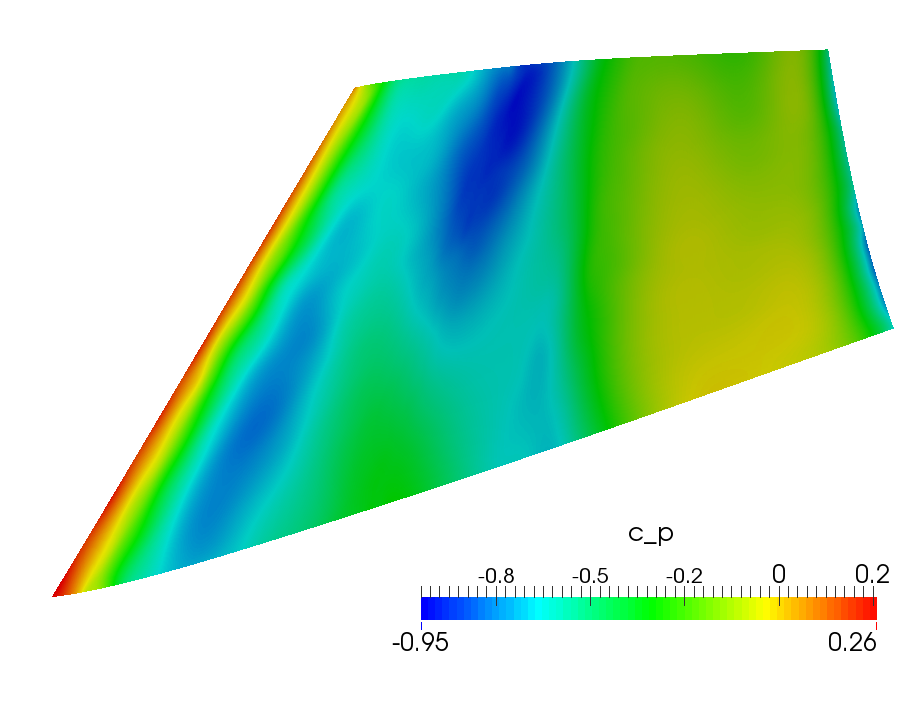}
\end{subfigure}%
\caption{\label{fig:CP_est_full} Original data (left) and smoothed estimated pressure coefficient $\hat{C}^h_p$, with $N=824$ data points and $N^h= 584$ basis functions (right).}
\end{figure}

We now assess the quality of IGS under the hypothesis that less data points than those effectively available can be used. Indeed, in industrial applications, the experimental measurements can be quite expensive and it could be impractical to have a large number of sampling points (in this case, pressure probes), and the pressure field can be measurable in fewer points than with numerical simulations. To assess the validity of IGS when few data points are available, we compute the aerodynamic force $\mathbf{\hat{F}}^h$ estimated from a subset of points over the winglet. We compare the results in terms of the aerodynamic force, which is the quantity of interest,  by using fewer data points and different number of basis functions with respect to the reference setting corresponding to $N^h=584$ and $N=824$. We compare the forces in terms of the relative difference of the modules (in \%), that is $\|\mathbf{\hat{F}}^h-\mathbf{\hat{F}}^h_{\text{ref}}\|_2 \slash \|\mathbf{\hat{F}}^h_{\text{ref}}\|_2  $ and of the direction (angle in degrees). The results are reported in Table \ref{tab:rel.diff.forces}.
We first use all the $N=824$ available points depicted in the Figure~\ref{fig:CP_est_full} (left) to estimate the pressure coefficient field. The estimation of the reference force is given by $\mathbf{\hat{F}}^h_{\text{ref}}= ( 0.5080, 49.40, -26.03)^T$ kN, which is in line with the expectations \dblue{in terms of direction and magnitude, according to the CFD simulations.}

\dblue{We then use only $N=618,412, 206$, and $103$ points and NURBS spaces of dimensions $N^h=584,989$, and $2079$.  The data sets with $N=103, 206$ and $412$ have been built using one, two and four over eight of the points in the original sequence, respectively; for $N=618$, the points are the complementary of those in the set with $N=206$.} In Figure~\ref{fig:soar_mult_CP_fields}, we report the estimated pressure coefficient fields on the surface computed using IGS. The choice of the smoothing parameter is done via minimization of the GCV criterion. The estimated functions are very similar when the number of points is relatively large ($N=824, 618,$ and $412$), as shown in Figure \ref{fig:soar_mult_CP_fields}. The relative change in the estimated force is quite small, as we can observe on Table \ref{tab:rel.diff.forces}. IGS is able to accurately estimate the force even with a small subset of the original data. 
\begin{table}[t!]
\centering
\begin{tabular}{c|c|c|c}
$N^h$ & $N$ & difference of direction (in degrees)& $\|\mathbf{\hat{F}}^h-\mathbf{\hat{F}}^h_{\text{ref}}\| \slash \|\mathbf{\hat{F}}^h_{\text{ref}}\|  $(in \%)\\
\hline
$584$ &  $103$ & 1.007 	& 3.293	\\
$584$ & $ 206$ & 0.816	& 1.036	\\
$584$ & $412$ & 0.084 	& 0.745\\
$584$ &  $618$ & 0.001  & 0.003	\\
\hline
$989$ &  $412$ & 0.278 	& 0.864	\\
$ 989$ & $618$ & 0.061 	& 0.205	\\
$989$ & $824$ & 0.003 	& 0.117\\
\hline 
$2079$ & $824$ & 0.0113 	& 0.296	\\
\end{tabular}
\caption{\label{tab:rel.diff.forces} Relative difference between vector of forces with respect to the reference force $\mathbf{\hat{F}}^h_{\text{ref}}$, estimated in the case where $N^h= 584$ and $N=824$, in degrees (left) and in norm (right).
Direction (in degrees) and relative modulus change (in \%) of the estimated aerodynamic force $\mathbf{\hat{F}}^h$ obtained with different points $N$ and basis functions $N^h$ with respect to the reference force $\mathbf{\hat{F}}^h_{\text{ref}}$.}
\end{table}

\begin{figure}[htbpt]
\centering
\begin{subfigure}{.49\linewidth}
\includegraphics[scale=0.23]{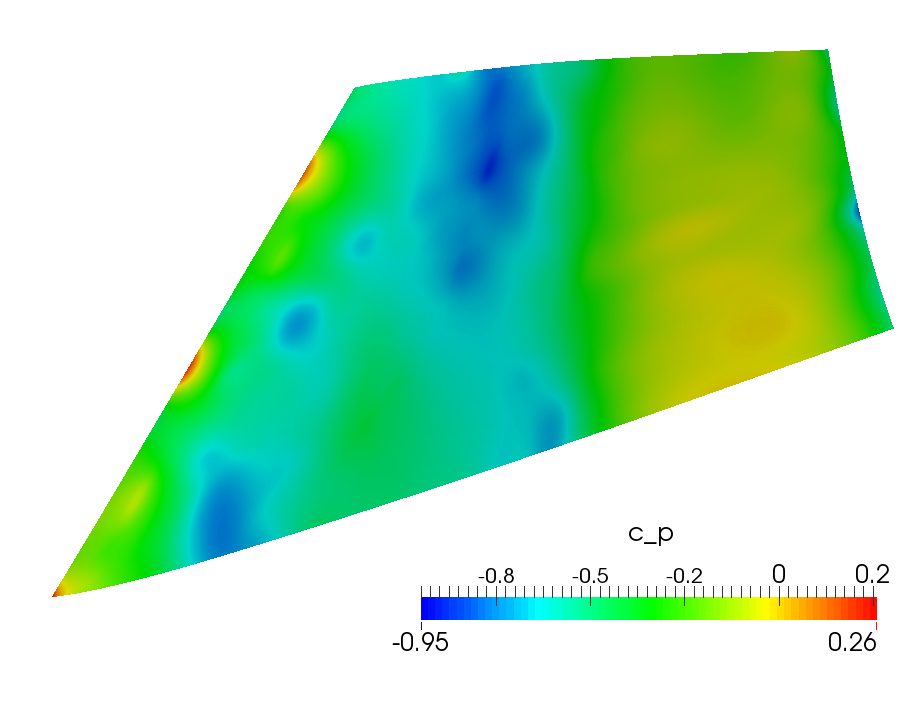} 
\subcaption{$N=103$ and $N^h=594$.}
\end{subfigure}
\begin{subfigure}{.49\linewidth}
\includegraphics[scale=0.24]{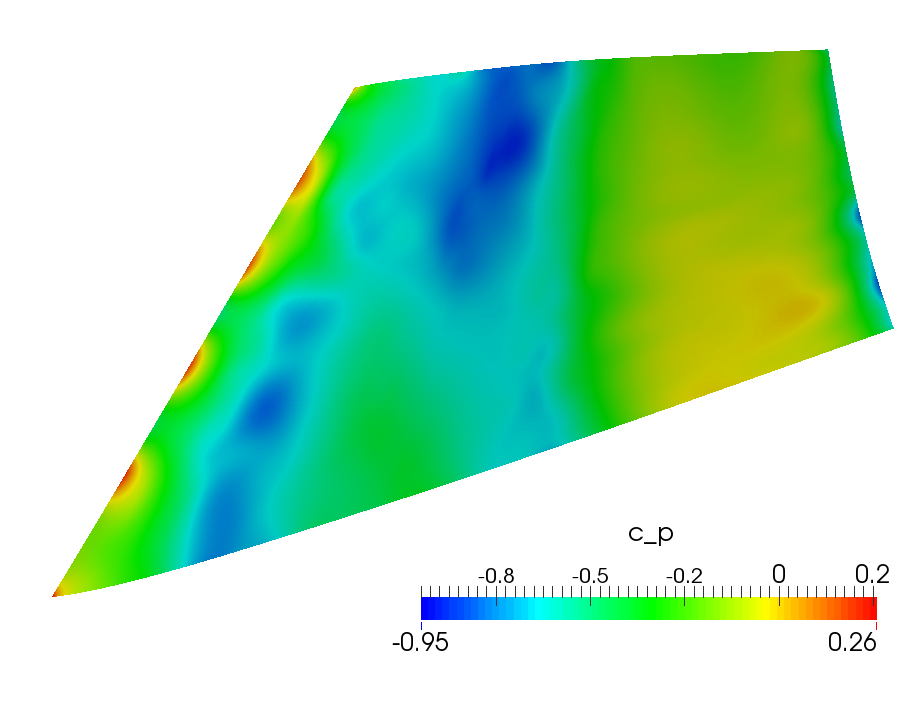}
\subcaption{$N=206$ and $N^h=594$.}
\end{subfigure}
\begin{subfigure}{.49\linewidth}
\includegraphics[scale=0.24]{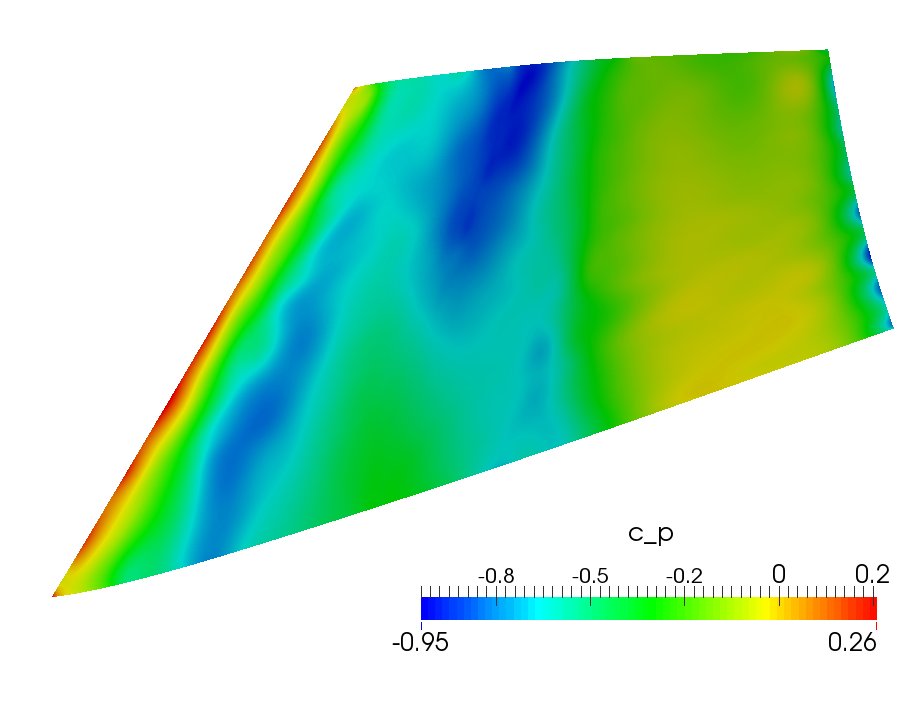}
\subcaption{$N=412$ and $N^h=594$.} 
\end{subfigure}
\begin{subfigure}{.49\linewidth}
\includegraphics[scale=0.24]{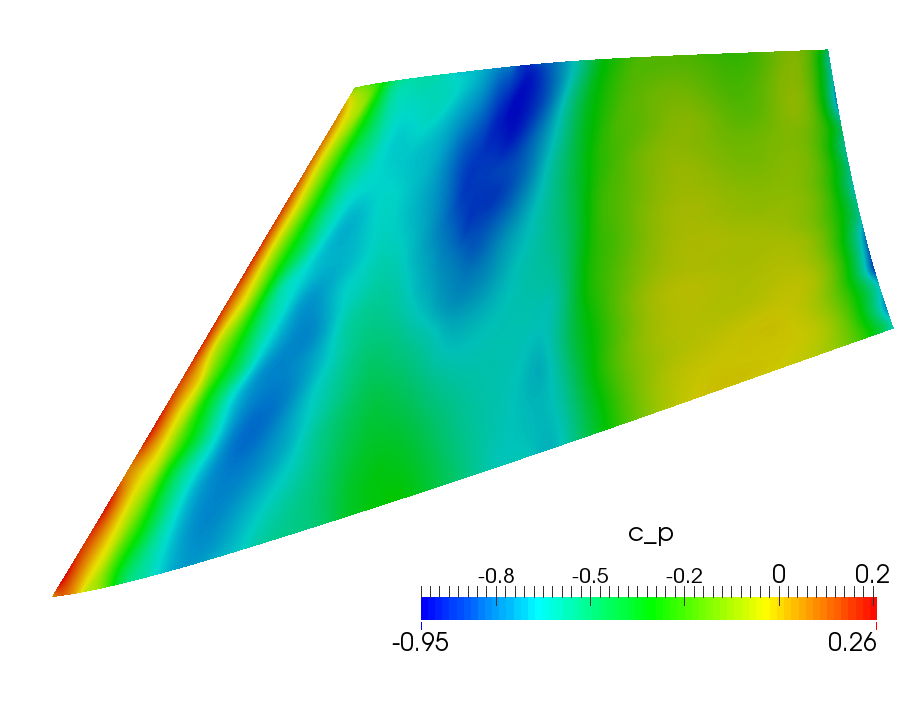}
\subcaption{$N=618$ and $N^h=594$.}
\end{subfigure}
\begin{subfigure}{.49\linewidth}
\includegraphics[scale=0.24]{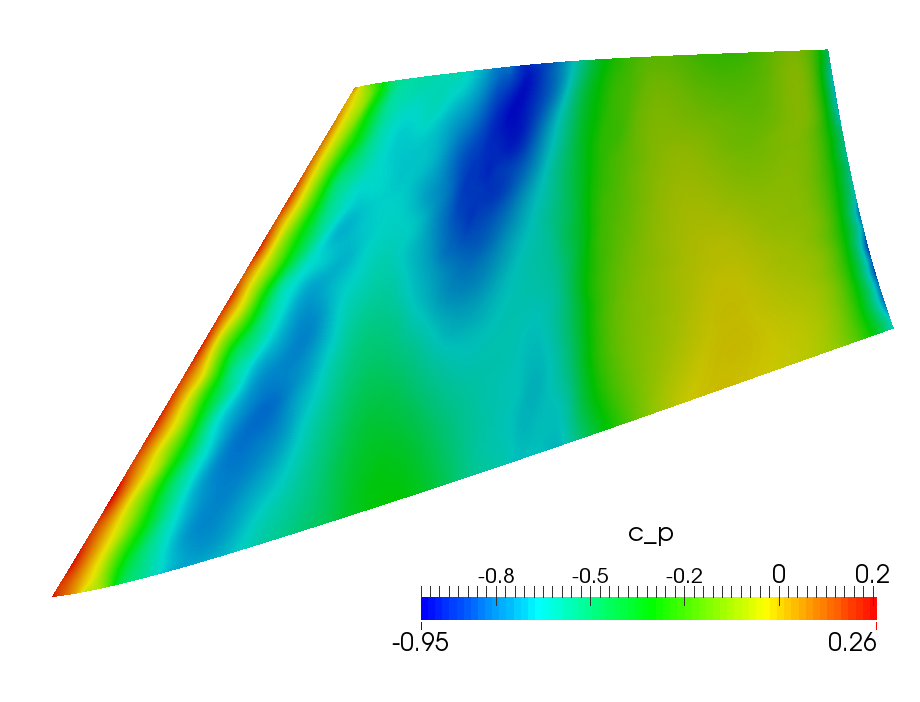} 
\subcaption{$N=824$ and $N^h=989$.}
\end{subfigure}
\begin{subfigure}{.49\linewidth}
\includegraphics[scale=0.24]{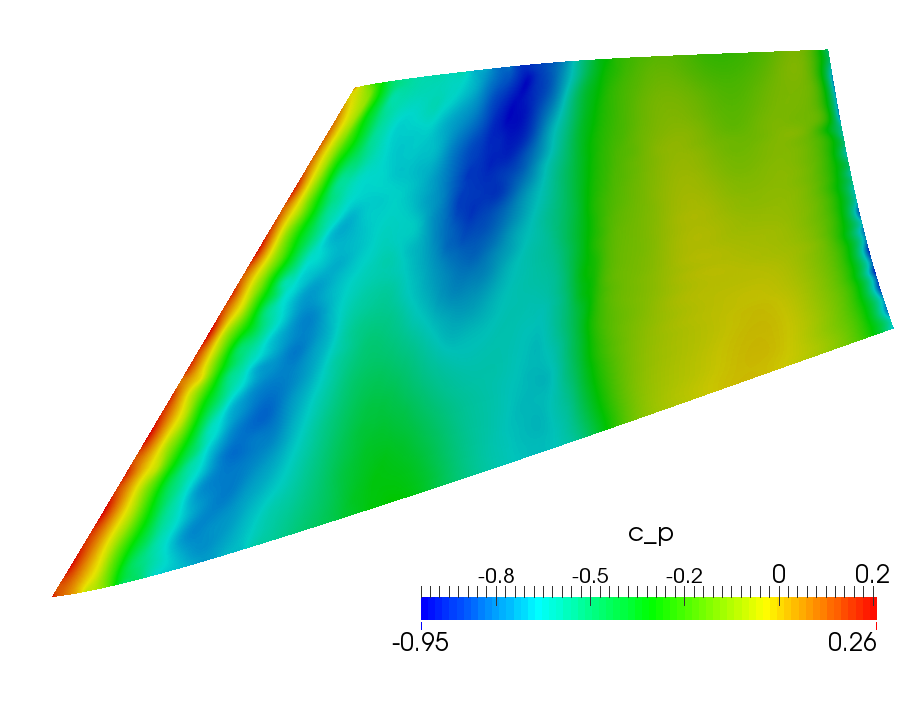}
\subcaption{$N=824$ and $N^h=2079$.} 
\end{subfigure}
\caption{\label{fig:soar_mult_CP_fields} Estimated pressure coefficient field $\hat{C}^h_p$ based on different number of data points $N$ and number of basis functions $N^h$.}
\end{figure}

\section{Conclusions and discussion}
In this paper, we proposed a methodology, IGS, to deal with functional data defined on surfaces described by NURBS, specifically to reconstruct the functions from noisy observations. The proposed IGS method has the potential of being widely applicable, in particular in industrial contexts where geometries are commonly defined by NURBS. Simulations indicate that IGS is comparable to other widely used methods, such as TPS, in cases where the latter is applicable. Moreover, IGS avoids the use of complex meshing procedures since the geometry of the surface is directly used as a data of the problem. IGS is also computationally efficient; this is due to the fact that the NURBS basis functions have compact support and thus the matrices involved in the computations are usually very sparse. These are very convenient features, especially if the number of basis functions is relatively large (e.g. in the case of a complex surface) and/or when there are many data points.

A differential penalization is equivalent to the assumption that the function to be estimated must be close to the kernel of the penalization operator \citep{Green93}. In the case where the observed physical phenomenon is driven by a known PDE, it would be convenient to use a penalization operator related to the PDE \citep{Azzimonti2}. Since IGS is based on IGA, it would be easy to implement other kind of penalization. IGS also allows the application of different kind of boundary conditions, which can be useful in practical applications \citep{Wood08,Sangalli13,Azzimonti1}. In addition, IGS can straightforwardly deal with data in one, two, or even three dimensions. Indeed, all the results presented here can be extended to any lower dimensional manifold defined by NURBS. Moreover, as shown in the application, functionals of an estimated field, as the aerodynamic force, can be easily computed. In addition, IGS is very flexible and allows local refinement of the basis functions, which can be needed when the distribution of the data points is not uniform. Finally, we remember that it is also possible to extend this model to take account of spatially varying covariates on the manifold, using a generalized additive framework.\vspace*{1cm}

\noindent \textbf{Acknowledgments} \\
M. W. and L. D. contributed equally to this paper. The authors aknowledge Prof. Alfio Quarteroni and Prof. Fabio Nobile (EPF Lausanne), Denis Devaud (ETH Zürich), and Lionel Wilhelm (EPF Lausanne) for very insightful discussions and suggestions. M. W. would like to thank Prof. Yves Till\'e (Université de Neuch\^atel) for his support.   

\section*{Bibliography}
\bibliography{isogeom}

\begin{thebibliography}{38}
\expandafter\ifx\csname natexlab\endcsname\relax\def\natexlab#1{#1}\fi
\expandafter\ifx\csname url\endcsname\relax
  \def\url#1{\texttt{#1}}\fi
\expandafter\ifx\csname urlprefix\endcsname\relax\def\urlprefix{URL }\fi

\bibitem[{Alfeld et~al.(1996)Alfeld, Neamtu, and Schumaker}]{Alfeld96}
Alfeld, P., Neamtu, M., Schumaker, L.~L., 1996. Fitting scattered data on
  sphere-like surfaces using spherical splines. Journal of Computational and
  Applied Mathematics 73~(1--2), 5--43.

\bibitem[{Anderson(2010)}]{Anderson10}
Anderson, J., 2010. Fundamentals of Aerodynamics. McGraw-Hill Education, New
  York.

\bibitem[{Azzimonti et~al.(2014)Azzimonti, Nobile, Sangalli, and
  Secchi}]{Azzimonti1}
Azzimonti, L., Nobile, F., Sangalli, L.~M., Secchi, P., 2014. Mixed finite
  elements for spatial regression with {PDE} penalization. SIAM/ASA Journal on
  Uncertainty Quantification 2~(1), 305--335.

\bibitem[{Azzimonti et~al.(2015)Azzimonti, Sangalli, Secchi, Domanin, and
  Nobile}]{Azzimonti2}
Azzimonti, L., Sangalli, L.~M., Secchi, P., Domanin, M., Nobile, F., 2015.
  Blood flow velocity field estimation via spatial regression with {PDE}
  penalization. Journal of the American Statistical Association 110~(511),
  1057--1071.

\bibitem[{Baramidze et~al.(2006)Baramidze, Lai, and Shum}]{Baramidze06}
Baramidze, V., Lai, M.~J., Shum, C.~K., 2006. Spherical splines for data
  interpolation and fitting. SIAM Journal on Scientific Computing 28~(1),
  241--259.

\bibitem[{Bartezzaghi et~al.(2015)Bartezzaghi, Dedè, and
  Quarteroni}]{Bartezzaghi15}
Bartezzaghi, A., Dedè, L., Quarteroni, A., 2015. Isogeometric analysis of high
  order partial differential equations on surfaces. Computer Methods in Applied
  Mechanics and Engineering 295, 446 -- 469.

\bibitem[{Beaubier et~al.(2014)Beaubier, Dufour, Hild, Roux, Lavernhe, and
  Lavernhe-Taillard}]{reviewer1}
Beaubier, B., Dufour, J.-E., Hild, F., Roux, S., Lavernhe, S.,
  Lavernhe-Taillard, K., 2014. {CAD}-based calibration and shape measurement
  with stereo{DIC}. Experimental Mechanics 54~(3), 329--341.

\bibitem[{Brezis(1999)}]{Brezis99}
Brezis, H., 1999. Anayse Fonctionnelle: Théorie et Applications. Dunod, Paris.

\bibitem[{Buja et~al.(1989)Buja, Hastie, and Tibshirani}]{Buja89}
Buja, A., Hastie, T.~J., Tibshirani, R.~J., 1989. Linear smoothers and additive
  models. The Annals of Statistics 17~(2), 453--510.

\bibitem[{Cottrell et~al.(2009)Cottrell, Hughes, and Bazilevs}]{Cottrell09}
Cottrell, J.~A., Hughes, T. J.~R., Bazilevs, Y., 2009. Isogeometric Analysis:
  Toward Integration of CAD and FEA. Wiley, Hoboken.

\bibitem[{Craven and Wahba(1978)}]{Craven78}
Craven, P., Wahba, G., 1978. Smoothing noisy data with spline functions.
  Numerische Mathematik 31~(4), 377--403.

\bibitem[{Dassi et~al.(2015)Dassi, Ettinger, Perotto, and Sangalli}]{Dassi15}
Dassi, F., Ettinger, B., Perotto, S., Sangalli, L.~M., 2015. A mesh
  simplification strategy for a spatial regression analysis over the cortical
  surface of the brain. Applied Numerical Mathematics 90, 111--131.

\bibitem[{Ded\`e and Quarteroni(2015)}]{Dede1}
Ded\`e, L., Quarteroni, A., 2015. Isogeometric {A}nalysis for second order
  {P}artial {D}ifferential {E}quations on surfaces. Computer Methods in Applied
  Mechanics and Engineering 284, 807-- 834.

\bibitem[{Duchamp and Stuetzle(2003)}]{Duchamp03}
Duchamp, T., Stuetzle, W., 2003. Spline smoothing on surfaces. Journal of
  Computational and Graphical Statistics 12~(2), 354--381.

\bibitem[{Duchon(1977)}]{Duchon77}
Duchon, J., 1977. Splines minimizing rotation-invariant semi-norms in {S}obolev
  spaces. In: Schempp, W., Zeller, K. (Eds.), Constructive Theory of Functions
  of Several Variables. Springer-Verlag, Berlin and Heidelberg.

\bibitem[{Dufour et~al.(2015)Dufour, Hild, and Roux}]{reviewer2}
Dufour, J.-E., Hild, F., Roux, S., 2015. Shape, displacement and mechanical
  properties from isogeometric multiview stereocorrelation. The Journal of
  Strain Analysis for Engineering Design 50~(7), 470--487.

\bibitem[{Ettinger et~al.(2015)Ettinger, Perotto, and Sangalli}]{Ettinger12}
Ettinger, B., Perotto, S., Sangalli, L.~M., 2015. Spatial regression models
  over two-dimensional manifolds. Biometrika.

\bibitem[{Green and Silverman(1993)}]{Green93}
Green, P.~J., Silverman, B.~W., 1993. Nonparametric Regression and Generalized
  Linear Models: A roughness penalty approach. Chapman \& Hall/CRC, Boca-Raton.

\bibitem[{Hastie and Tibshirani(1990)}]{Hastie90}
Hastie, T.~J., Tibshirani, R.~J., 1990. Generalized Additive Models. Chapman \&
  Hall, CRC Press, Boca-Raton.

\bibitem[{Hughes et~al.(2005)Hughes, Cottrell, and Bazilevs}]{Hughes05}
Hughes, T. J.~R., Cottrell, J.~A., Bazilevs, Y., 2005. Isogeometric {A}nalysis:
  {CAD}, {F}inite {E}lements, {NURBS}, exact geometry and mesh refinement.
  Computer Methods in Applied Mechanics and Engineering 194~(39-41), 4135 --
  4195.

\bibitem[{Lai and Schumaker(2007)}]{Lai07}
Lai, M.~J., Schumaker, L.~L., 2007. Spline Functions on Triangulations.
  Cambridge University Press, Cambridge.

\bibitem[{Marra and Wood(2012)}]{Marra12}
Marra, G., Wood, S.~N., 2012. Coverage properties of confidence intervals for
  generalized additive models components. Scandinavian Journal of Statistics
  39, 53--–74.

\bibitem[{Nocedal and Wright(1999)}]{Nocedal99}
Nocedal, J., Wright, S.~J., 1999. Numerical Optimization. Springer-Verlag, New
  York.

\bibitem[{Piegl and Tiller(1997)}]{Piegl97}
Piegl, L., Tiller, W., 1997. The NURBS Book. Springer-Verlag, New York.

\bibitem[{Quarteroni(2015)}]{Quarteroni09}
Quarteroni, A., 2015. Numerical Models for Differential Problems.
  Springer-Verlag, Milan.

\bibitem[{Quarteroni and Valli(1994)}]{Quarteroni94}
Quarteroni, A., Valli, A., 1994. Numerical Approximation of Partial
  Differential Equations. Springer-Verlag, Heidelberg.

\bibitem[{Ramsay and Silverman(2005)}]{Silverman05}
Ramsay, J.~O., Silverman, B.~W., 2005. Functional Data Analysis.
  Springer-Verlag, New York.

\bibitem[{Ramsay(2002)}]{Ramsey02}
Ramsay, T.~O., 2002. Spline smoothing over difficult regions. Journal of the
  Royal Statistical Society: Series B (Statistical Methodology) 64~(2),
  307--319.

\bibitem[{Sangalli et~al.(2013)Sangalli, Ramsay, and Ramsay}]{Sangalli13}
Sangalli, L.~M., Ramsay, J.~O., Ramsay, T.~O., 2013. Spatial spline regression
  models. Journal of the Royal Statistical Society: Series B (Statistical
  Methodology) 75, 681--703.

\bibitem[{Stoker(1989)}]{Stoker89}
Stoker, J.~J., 1989. Differential Geometry. Wiley, Hoboken.

\bibitem[{Tagliabue et~al.(2014)Tagliabue, Ded\`e, and Quarteroni}]{Dede2}
Tagliabue, A., Ded\`e, L., Quarteroni, A., 2014. Isogeometric {A}nalysis and
  error estimates for high order {P}artial {D}ifferential {E}quations in fluid
  dynamics. Computers \& Fluids 102, $277$--$303$.

\bibitem[{Wahba(1981)}]{Wahba81}
Wahba, G., 1981. Spline interpolation and smoothing on the sphere. SIAM Journal
  on Scientific and Statistical Computing 2~(1), 5--16.

\bibitem[{Wahba(1990)}]{Wahba90}
Wahba, G., 1990. Spline Models for Observational Data. Society for Industrial
  and Applied Mathematics, Philadelphia.

\bibitem[{Wilhelm(2013)}]{Wilhelm13}
Wilhelm, M., 2013. Generalized spatial regression with differential
  penalization. Master's thesis, {\'{E}}cole {P}olytechnique
  {F}{\'{e}}d{\'{e}}rale de Lausanne.
\newline\urlprefix\url{//infoscience.epfl.ch/record/188219}

\bibitem[{Wood(2006)}]{Wood06}
Wood, S.~N., 2006. Generalized additive models: an introduction with
  application in R. Chapman \& Hall/CRC, Boca-Raton.

\bibitem[{Wood(2011)}]{Wood11}
Wood, S.~N., 2011. Fast stable restricted maximum likelihood and marginal
  likelihood estimation of semiparametric generalized linear models. Journal of
  the Royal Statistical Society: Series B (Statistical Methodology) 73~(1),
  3--36.

\bibitem[{Wood(2015)}]{mgcv}
Wood, S.~N., 2015. \texttt{mgcv}: Mixed GAM Computation Vehicle with
  GCV/AIC/REML smoothness estimation. R package version on 1.8-6.

\bibitem[{Wood et~al.(2008)Wood, Bravington, and Hedley}]{Wood08}
Wood, S.~N., Bravington, M.~V., Hedley, S.~L., 2008. Soap film smoothing.
  Journal of the Royal Statistical Society: Series B (Statistical Methodology)
  70~(5), 931--955.

\end{thebibliography}
\bibliographystyle{elsarticle-harv}
\addcontentsline{toc}{chapter}{Bibliography}
\end{document}